\documentclass[prb, twocolumn, showpacs, preprintnumbers, amsmath, amssymb,
  floats, floatfix, superscriptaddress]{revtex4-1}
\usepackage[english]{babel}
\usepackage[table]{xcolor}
\usepackage{amsmath,amssymb,graphicx,hyperref,latexsym}

\newenvironment{proof}{\noindent\textbf{Proof.}}{\hspace*{\fill}$\Box$\medskip}
\newtheorem{theorem}{Property}
\let\oldbullet\bullet
\renewcommand{\bullet}{{\oldbullet}}
\let\oldcirc\circ
\renewcommand{\circ}{{\oldcirc}}
\newcommand\bigforall{\mbox{\Large $\mathsurround0pt\forall$}}

\begin{document}

\title{Fermionic phases and their transitions induced by competing finite-range interactions}

\author{M.\ Szyniszewski}
\email{mszynisz@gmail.com}
\affiliation{Department of Physics, Lancaster University, Lancaster LA1 4YB,
  United Kingdom}

\author{H.\ Schomerus}
\affiliation{Department of Physics, Lancaster University, Lancaster LA1 4YB,
  United Kingdom}

\date{\today}

\begin{abstract}
  We identify ground states of one-dimensional fermionic systems subject to
  competing repulsive interactions of finite range, and provide phenomenological
  and fundamental signatures of these phases and their transitions.
  Commensurable particle densities admit multiple competing charge-ordered
  insulating states with various periodicities and internal structure. Our
  reference point are systems with interaction range $p=2$, where phase
  transitions between these charge-ordered configurations are known to be
  mediated by liquid and bond-ordered phases. For increased interaction range
  $p=4$, we find that the phase transitions can also appear to be abrupt, as
  well as being mediated by re-emergent ordered phases that cross over into
  liquid behavior. These considerations are underpinned by a classification of
  the competing charge-ordered states in the atomic limit for varying
  interaction range at the principal commensurable particle densities. We also
  consider the effects of disorder, leading to fragmentization of the ordered
  phases and localization of the liquid phases.
\end{abstract}

{\maketitle}

\section{Introduction}

One-dimensional fermionic physics substantially differs from its
higher-dimensional counterparts. Usual descriptions of interactions, such as
Fermi liquid theory, break down, \cite{Peierls1955, Voit1995} which results in
an absence of quasi-particle excitations in the system. Under many
circumstances, the appropriate theory describing these systems is the
Tomonaga-Luttinger liquid.\cite{Tomonaga1950, Luttinger1963, Mattis1965,
Haldane1981, Haldane1981prl} However, in the presence of repulsive interactions
and at commensurable particle densities the system can form charge-ordered
phases,\cite{Peierls1930, Mott1937, Mott1949, Peierls1955, Bari1971, Gruner1988}
in which it displays insulating properties.

The nature of the quantum phase transition between liquid and charge-ordered
phases has been uncovered in detailed theoretical studies.\cite{Frohlich1954,
Kuper1955, Hubbard1963, Fourcade1984, Giamarchi1997} The main factor that drives
the phase transition is the competition between the kinetic energy and the
interaction terms that can order the system. In particular, lattice fermion
models exhibit an interplay of short-range kinetic quantum fluctuations arising
from the uncertainty principle and repulsive finite-range interactions that
cause the insulating phases. \cite{Gomez-Santos1993, Schmitteckert2004,
Mishra2011, Dodaro2017} These studies have revealed various transition
scenarios, including the emergence of a `strange metallic phase' that was later
identified to be of bond order \cite{Voit1992, Tsiper1997, Nakamura1999,
Zhuravlev2000, Sengupta2002, Duan2011} -- a dimerized phase with alternating
bond strengths and without charge ordering. To determine these quantum phase
transitions precisely, various methods have been proposed, such as investigation
of the ground-state curvature, \cite{Schmitteckert2004} structure factors,
\cite{Mishra2011, Duan2011} bond entropy, \cite{Molina2007} and scaling of the
gap. \cite{Mishra2011} Decaying behaviour of the correlation functions was also
found to be a distinguishing feature for the aforementioned quantum phases.
\cite{Hallberg1990, Duan2011} Since spinless fermions in one dimension are
equivalent to hard-core bosons, \cite{JordanWigner1928, Lieb1961, Katsura1962,
Haldane1979} these observations are also of great interest for analogous bosonic
systems. \cite{Fisher1989} Recent advancements in optical lattices have indeed
made it possible to engineer condensed matter systems \cite{Lubasch2011,
Cazalilla2011} that allow to directly observe the liquid-to-insulator
transition. \cite{Buchler2003, Liu2005, Haller2010}

Overall, however, the understanding of these transitions is still restricted to
a small number of relatively simple and mutually well compatible charge-ordered
states. As the range of interactions increases, the variety of competing
charge-ordered states increases rapidly. This situation raises a number of
unresolved questions. On one hand, the transitions between these phases may
proliferate as well, and could squeeze out the insulating behavior. On the other
hand, the liquid phases could be suppressed depending on the complexity of the
charge configurations of competing states. Moreover, the bond-ordered phases may
survive the introduction of additional interactions, or, completely new
transition scenarios could arise.

In this paper, we address these questions within a model that exhibits multiple
transitions between a variety of ordered states of varying compatibility. This
leads to a rich phase diagram where we can explore how insulating phases survive
when the range of the repulsive interactions is increased, and which transitions
between different insulating phases can occur. The competing interactions of
finite range give rise to a multitude of charge-ordered phases, which we
systematically classify at the principal critical particle densities. We then
investigate a hierarchy of signatures that characterize the phases and their
transitions at the fundamental and phenomenological level.

Inspection of the atomic limit where the kinetic energy term vanishes allows us
to systematically identify the candidate charge-ordered phases. In this limit,
the phase transitions are sharp and are only driven by considerations of the
interaction energy, while liquid phases are absent. The carrier mobility at a
finite kinetic energy gives scope for liquid behavior that can intervene between
the charge-ordered states. The consequences are investigated numerically using
the infinite-system density matrix renormalization group (iDMRG) approach with a
ground state represented as an infinite matrix product state
(iMPS).\cite{White1992, White1993, Schollwock2005, McCulloch2007, McCulloch2008,
Schollwock2011, MPToolkit} Many of the resulting features are already visible on
the phenomenological level, as we demonstrate for the experimentally accessible
kinetic energy density and the bond-order parameter, which display
characteristic discontinuities at many (but not all) of the phase transitions
identified on a more fundamental level. For the latter we employ density-density
correlation functions that capture long-range charge ordering, as well as the
bipartite entanglement entropy, which displays characteristic scaling in
critical phases.

As our main findings, we observe that depending on the compatibility of the
ordered states, the liquid phase can be strongly suppressed to the extent that
the transition appears to remain directly insulator-to-insulator. Furthermore,
we uncover the re-emergence of simple charge-ordered phases that mediate the
transition between more complicated ones, and exhibit a crossover to liquid
behavior at one of the phase boundaries.

We also consider the implications of disorder, which affects the charge ordering
by inducing fragmentization and further suppresses liquid behavior via
localization. At large disorder strength the system displays the characteristics
of a universal fragmented insulating phase.

This paper is organized as follows. In Sec.\ \ref{sec:model} we present the
model, method, and further background for this work. The charge-ordered phases
of the model in the atomic limit are described in Sec.\ \ref{sec:atomic}.
Section~\ref{sec:finitet} discusses the consequences of a finite kinetic energy,
where the emergent liquid behavior is supplemented by the direct and
crossover-mediated transitions between charge-ordered phases described above.
The disordered system is studied in Sec.\ \ref{sec:disorder}, and our
conclusions are given in  Sec.\ \ref{sec:conclusions}. The Appendix contains
details about the classification of charge-ordered phases for principal critical
particle densities in the atomic limit.

\section{Model, background and methods}\label{sec:model}

We base our investigations on a model of spinless fermionic particles that move
on a one-dimensional chain of size $L$ and are interacting through a
finite-range repulsive potential of maximal range $p$. The disorder-free
Hamiltonian of this model is given by\cite{Gomez-Santos1993}
\begin{equation}
  H = - t \sum_{i = 1}^L \left( c_i^{\dag} c_{i + 1} + \text{h.c.} \right) +
  \sum_{i = 1}^L \sum_{m = 1}^p U_m n_i n_{i + m}, \label{eq:Ham}
\end{equation}
where $c_i^{\dagger}, c_i$ are fermionic creation and annihilation operators on
site $i=1,\ldots,L$, $n_i = c_i^{\dagger} c_i$ are the corresponding
particle-number operators, $t$ determines the kinetic energy, and $U_m$ is the
interaction energy between two particles that are $m\leq p$ sites apart. All
interactions are assumed to be repulsive ($U_m > 0$). While only the ratios
$U_m/t$ matter for the properties of the system, we will treat these scales
independently as this facilitates the discussion of the atomic limit ($t\to 0$).
The particle density is denoted as $Q = L^{-1}\langle \sum_i n_i \rangle$.
Disorder can be included via a term
\begin{equation}
  H_{\text{dis}} = \sum_{i = 1}^L h_i \left( n_i - \frac{1}{2} \right),
  \label{eq:Hdis}
\end{equation}
with uniformly distributed random potentials $h_i \in [- W, W]$ at disorder
strength $W$.

In the seminal Ref.~{\onlinecite{Gomez-Santos1993}}, the potential energy is
strictly convex ($U_{m+1}+U_{m-1}>2U_m$), which assures that there is at most
one insulating phase for any given particle density $Q$ in the system.  These
phases can then be investigated assuming a hierarchy of well-separated energy
scales $t \ll \cdots \ll U_3 \ll U_2 \ll U_1$, hence close to the atomic limit.
Under these conditions the system is found to sustain a charge-ordered
insulating phase at any commensurable particle density
\begin{equation}
  Q_m = 1 / m, \quad m = p + 1, p, \ldots, 2,
  \label{QC}
\end{equation}
while otherwise the system behaves as a Luttinger liquid.

A distinctively more non-trivial behavior can be encountered at these critical
densities if the interaction potential is not convex, so that several
charge-ordered states can compete at the same commensurable particle density.
The convexity condition was abandoned in previous studies of the case $p=2$,
where the system is also known as the $t$-$V$-$V'$ model. {\cite{Emery1988,
Hallberg1990, Poilblanc1997, Franco2003, Schmitteckert2004, Seo2006, Molina2007,
Mishra2011}} This revealed that two charge-ordered states can compete at half
filling, and that the transition between these phases is mediated by a liquid
phase and bond-ordered phases. In this paper, we explore this competition for
the much broader range of charge orderings that occur at larger values of the
interaction range $p$. In the $\{ U_m \}$ phase diagram, this gives rise to
multiple instances of charge-ordered phases separated by intervening states that
mediate their transition, which are the main focus of this work.

The scene will be set by the analytical classification of charge-ordered phases
in the atomic limit $t\to 0$, while the consequences of a finite kinetic energy
are investigated numerically. We adopt a density-matrix renormalization group
approach \cite{White1992, White1993, Schollwock2005} based on a tensor-network
formulation \cite{Perez-Garcia2007, Verstraete2008, Perez-Garcia2008} where the
target states are represented by matrix product states, \cite{McCulloch2007,
McCulloch2008, Schollwock2013, Schollwock2011, Hubig2015} and utilize for this
the Matrix Product Toolkit {\cite{MPToolkit}} code together with our
implementation of the Hamiltonian~\eqref{eq:Ham}. This approach circumvents,
e.g., the restriction to small system sizes encountered in exact diagonalization
and the fermionic sign problem encountered in Quantum Monte Carlo approaches.
\cite{Sandvik2010} To investigate the ground state near the thermodynamic limit,
the desired state of the system is represented as an iMPS, which accounts for an
infinite number of unit cells. Note that iDMRG used in the iMPS context is
different from the infinite-size algorithm of DMRG in the context of finite
systems.\cite{McCulloch2007, McCulloch2008} During each step of iDMRG the
filling is kept at $Q$ (the U(1) symmetry is preserved -- for details, including
a discussion of spontaneous breaking of discrete symmetries, see
Ref.~\onlinecite{McCulloch2008}).

The iMPS unit cell size is chosen to make sure that the system is commensurable
with all possible insulating phases determined in the atomic limit, whereby we
avoid the frustration of any relevant charge-ordered state. Specifically, for a
half-filled ($Q=1/2$) system with $p = 2$, the possible charge-ordered states
have periods two and four, so that we choose a unit cell of 4 sites. For $Q=1/2$
but $p=4$, the phases are far more richer, which requires a unit cell of size
24. The maximal number of saved states (bond dimension $\chi$) in the iDMRG
procedure is 1000.

Using these tools, we characterize the phases by a set of complementary
signatures. For the most phenomenological description we consider the kinetic
energy density
\begin{equation}
  T = \frac{1}{L} \left\langle \sum_{i = 1}^L \left( c_i^{\dag} c_{i + 1} +
  \text{h.c.} \right) \right\rangle.
 \label{eq:t}
\end{equation}
This is a single-particle observable that probes the particle mobility between
neighboring sites and can, in principle, be assessed in atom-optical experiments
by time-of-flight measurements of atoms released from the optical lattice.

The extent of bond-order is addressed by the order parameter,\cite{Mishra2011}
\begin{equation}
  O_{\mathrm{BO}} = \frac{1}{L} \left\langle \sum_{i = 1}^L (-1)^i\left( c_i^{\dag} c_{i + 1} +
  \text{h.c.} \right) \right\rangle,
 \label{eq:obo}
\end{equation}
which constitutes a staggered version of the kinetic energy density. This
parameter measures the amount of the dimerization in the system. We report its
absolute value, which is invariant under the translation of the measured state.

We note that $O_\mathrm{BO}$ can also be finite in certain charge-ordered
states. This ambiguity is resolved by supplementing this quantity with
additional information. The required detailed insight into the charge ordering
is provided by the density-density correlation functions
\begin{equation}
  N_m = \frac{1}{L} \left\langle \sum_{i = 1}^L n_i n_{i + m}
  \right\rangle ,
  \label{eq:nm}
\end{equation}
which probe the ordering of particles that are $m$ sites apart, and allow to
further discriminate charge-ordered from bond-ordered and liquid phases. To
describe the long-range effects in the system, we exploit that $N_m$ develops an
oscillating behavior in $m$. More precisely, we observe that the limit
\begin{equation}
  \lim_{k \to \infty} N_{m + kP} = N_m^{\infty}, \quad m=1,\ldots,P
  \label{eq:nminf}
\end{equation}
exists, where $P$ is the unit-cell size of the charge order. We call
$N_m^{\infty}$ the extrapolated density-density correlator. This quantity
describes the long-range charge correlations in the system.

Finally, on the most fundamental level we characterize the quantum phases and
transitions by the scaling of the bipartite von Neumann entanglement entropy
\begin{equation}
  S = - \text{tr} (\rho_A \log_2 \rho_A),
\end{equation}
where $\rho_A$ is the reduced density matrix of a subchain $A$. Away from
quantum-critical behavior, $S$ scales as the system's boundary (the well-known
area law) \cite{Vidal2003, Latorre2004, Eisert2010}, and therefore converges
with increasing bond dimension $\chi$ in the charge-ordered and bond-ordered
phases. If the system is critical, the entropy is expected to increase
logarithmically with $\chi$ \cite{Calabrese2004, Tagliacozzo2008}, which in our
investigation occurs at phase transitions and in the liquid phase. In the iDMRG
algorithm, the entropy is calculated during each step {\cite{Schollwock2011}},
and therefore requires no additional computational cost.

\section{Atomic limit}\label{sec:atomic}

\begin{table}[t]
  \centering
  \caption{Number of distinct charge-ordered insulating phases in the atomic
    limit $t\to 0$ of the model \eqref{eq:Ham}, for different interaction ranges
    $p$ and commensurable particle densities $Q$. For details of the
    construction see Appendix \ref{sec:app1}. \label{tab:insulatingphases}}
  \begin{tabular}{ll|ccccccc}
    \hline \hline
    &  & $Q =$ &\multicolumn{6}{c}{$\longleftarrow$} \\
    &  & 1/2 & 1/3 & 1/4 & 1/5 & 1/6 & 1/7 & $\cdots$\\
    \hline
    $p =$ & 1 & 1 &  &  &  &  &  &  \\
    & 2 & 2 & 1 &  &  &  &  &  \\
    & 3 & 3 & 3 & 1 &  &  &  &  \\
    ${\downarrow}$ & 4 & 5 & 7 & 4 & 1 &  &  &  \\
    & 5 & 8 & 12 & 7 & 5 & 1 &  &  \\
    & 6 & 12 & $\geqslant 63$ & $\geqslant 23$ & 9 & 6 & 1 &  \\
    & $\vdots$ &  &  &  &  &  &  & $\ddots$\\
    \hline \hline
  \end{tabular}
\end{table}

To prepare the investigation of quantum phase transitions between the insulating
phases of different charge order, we first inspect the atomic limit of $t \to 0$
at different values of $p$ and critical densities $Q=Q_m$ [see Eq.~\eqref{QC}].
In this limit we can identify the distinct charge-ordered phases by purely
combinatorial energetic considerations.

The most trivial case occurs at density $Q = Q_{p + 1}$, where the fermions can
be spread out evenly across the system so that they are outside of the range of
their interactions. This then defines a universal ground-state with vanishing
energy, which is $(p + 1)$-fold degenerate.

For density $Q = Q_p$, the ground state can constitute any one of $p$ distinct
candidate phases, which we enumerate by an index $\alpha=1,\ldots,p$. As shown
in Appendix \ref{sec:insulating-construction}, these consist of $N / (p - \alpha
+ 1)$ blocks of a single fermion accompanied by $(\alpha - 1)$ empty sites, and
$N (p - \alpha) / (p - \alpha + 1)$ blocks of a single fermion accompanied by
$p$ empty sites, where $N$ is the number of fermions in the considered segment.
The competition between these phases is governed by their energy $E_\alpha =
NU_{\alpha} / (p - \alpha + 1)$, so that the ground-state phase $\alpha$ is
selected by the condition
\begin{equation}
  U_{\alpha} < \frac{p - \alpha +
  1}{p - \beta + 1} U_{\beta}, \quad \beta\neq \alpha.
  \label{eq:udom}
\end{equation}
An important example is the case $p = 2$, $Q=Q_2=1/2$, which corresponds to the
$t$-$V$-$V'$ model at half filling studied in
Refs.~\onlinecite{Schmitteckert2004, Mishra2011}. The phase diagram then
consists of two phases: one with a ground-state unit cell of $(\bullet \circ)$,
where $\bullet$ is an occupied site and $\circ$ is an unoccupied site; and one
with a unit cell of $(\bullet \bullet \circ \circ)$. The two phases have energy
densities of $U_2 / 2$ and $U_1 / 4$, respectively, and the phase transition
occurs along the $U_1 = 2 U_2$ line.

\begin{table}[t]
  \centering
  \caption{(Color online) Ground-state (GS) unit cells and their energies in the
  atomic limit of half-filled systems  ($Q=1/2$) with interaction range $p=2$
  and $p=4$. In the pictorial representations of the unit cells, $\bullet$
  denotes an occupied site and $\circ$ denotes an empty site. The degeneracy $f$
  accounts for the translational freedom of these phases. The colors designate
  their position in the phase diagrams of Fig.~\ref{fig:exampleCDW}.
  \label{tab:exampleCDW}}
  \begin{tabular}{cccc}
    \hline \hline
    GS unit cell & Energy density & $f$ & \\
\hline
    \multicolumn{4}{c}{$p = 2, Q = 1/2$}\\
    \hline
    $\bullet \circ$ & $U_2/2$ & 2 &
    {\color[HTML]{00AAFF}$\blacksquare$}\\
    $\bullet \bullet \circ \circ$ & $U_1/4$ & 4 &
    {\color[HTML]{55FF00}$\blacksquare$}\\
    \hline
    \multicolumn{4}{c}{$p = 4, Q = 1/2$}\\
    \hline
    $\bullet \circ$ & $(U_2 + U_4)/2$ & 2 &
    {\color[HTML]{00AAFF}$\blacksquare$}\\
    $\bullet \bullet \circ \circ$ & $(U_1 + U_3 + 2 U_4)/4$ & 4 &
    {\color[HTML]{55FF00}$\blacksquare$}\\
    $\bullet \bullet \bullet \circ \circ \circ$ & $(2 U_1 + U_2 +
    U_4)/6$ & 6 & {\color[HTML]{FFFF00}$\blacksquare$}\\
    $\bullet \bullet \bullet \bullet \circ \circ \circ \circ$ & $
    (3 U_1 + 2 U_2 + U_3)/8$ & 8 & {\color[HTML]{FF0000}$\blacksquare$}\\
    $\bullet  \circ\circ \bullet \circ \bullet \bullet \circ$ & $
    (U_1 + 2 U_2 + 3 U_3)/8$ & 8 & {\color[HTML]{800080}$\blacksquare$}\\
    \hline \hline
  \end{tabular}
\end{table}

\begin{figure}[t]
  \includegraphics[width=0.9\columnwidth]{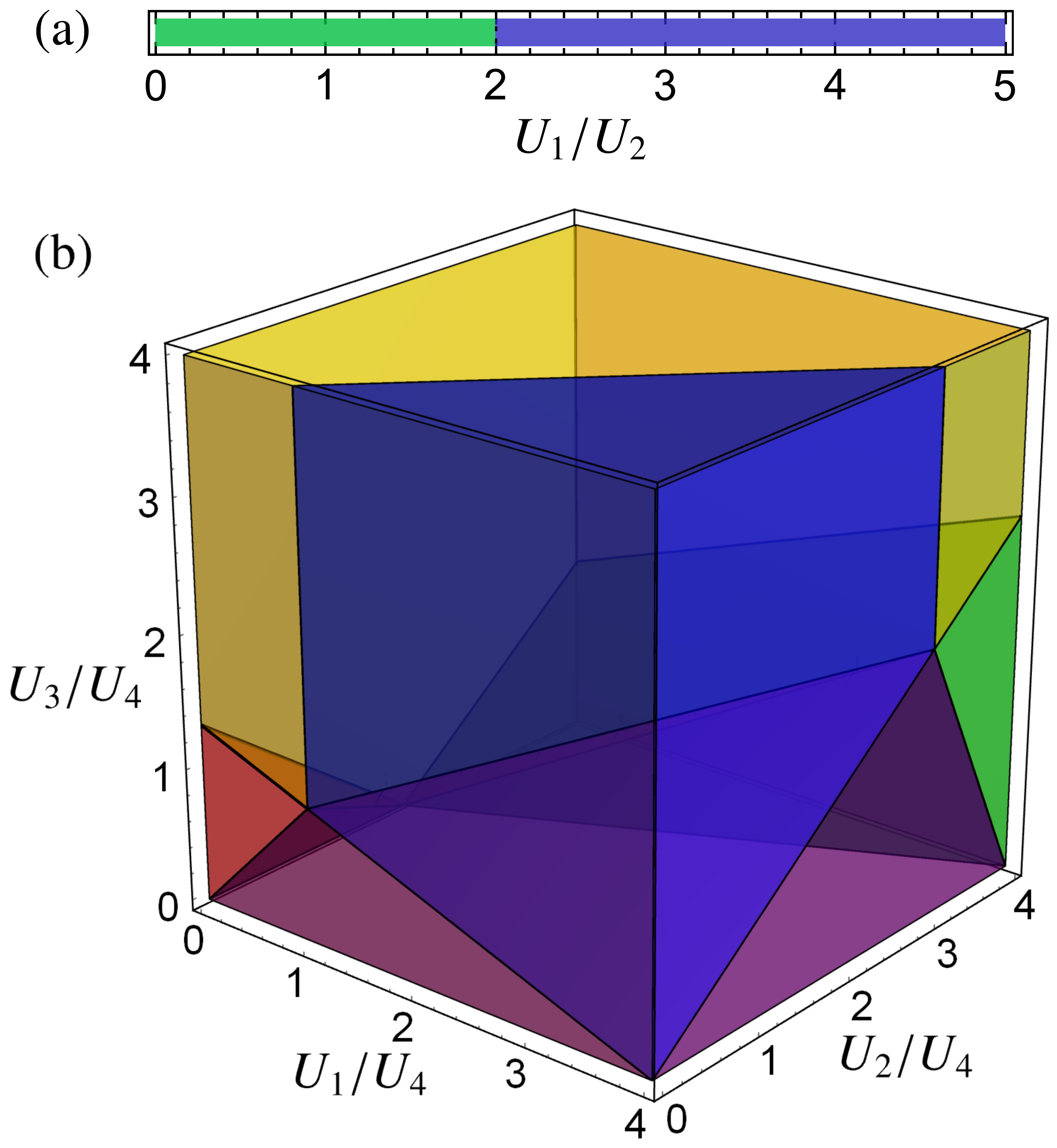}
  \caption{(Color online) Phase diagrams of charge-ordered  ground states in the
  atomic limit of half-filled systems  ($Q=1/2$) with interaction range $p=2$
  (a) and $p=4$ (b).  In  both cases we exploit the freedom to choose $U_p$ as
  the energy scale. The colors indicate the phases listed in
  Tab.~\ref{tab:exampleCDW}.\label{fig:exampleCDW}}
\end{figure}

At higher critical densities a much wider variety of possible charge orderings
emerges, whose competition can be assessed by a combinatorial analysis. Table
\ref{tab:insulatingphases} shows an overview of the number of insulating phases
that we could determine in the atomic limit of systems with $p\leq 6$. The
number of phases grows rapidly with the interaction range $p$ and depends
distinctively on the density $Q$. The detailed configurations of the
corresponding charge-ordered states are specified in Appendix
\ref{sec:insulatingPhasesList}. This reveals phases with highly intricate
internal structure, indicating that there are no simple rules governing the
ground-state properties of systems at high critical particle densities.

In the remainder of this paper we focus on the representative example of $p = 4,
Q=Q_2=1/2$ (hence again half filling). Table \ref{tab:exampleCDW} lists the
distinct charge-ordered phases for this case.  Although the first four phases
follow a relatively simple ordering pattern, the fifth phase displays a more
intricate internal structure. The corresponding phase diagram is shown in
Fig.~\ref{fig:exampleCDW}, which displays the phases in the space of interaction
parameters $U_1$, $U_2$ and  $U_3$, while $U_4$ serves as the energy scale.

\section{Finite kinetic energy}\label{sec:finitet}

In the atomic limit, the phase transitions between the charge-ordered phases are
abrupt. This situation changes at a finite kinetic energy $t\neq 0$, where the
transitions can be mediated by other phases, such as the liquid and bond-ordered
phases previously encountered in the $t$-$V$-$V'$ model ($p=2$). As the number
of competing phases increases rapidly with larger interaction range, one could
suspect that the phase space may be dominated by the transitions between these
phases, while the insulating phases are only present close to the atomic limit.
Thus, a large interaction range may imply the loss of insulating properties of
the system. Furthermore, it is per se unclear how the distinct internal
structure of the charge-ordered states affects the nature of the transitions. We
now explore these questions for the case $p=4$ at half filling, corresponding to
the competition of ordered phases listed in Table \ref{tab:exampleCDW}, and
contrast this case with $p=2$. Throughout most of this section, we set $t=1$ to
fix the unit of energy, and utilize the complementary signatures described in
Sect.~\ref{sec:model}.

For $p=2$, we find representative behavior by fixing $U_1=10$ while varying
$U_2$, which allows us to verify the consistency of our results with previous
studies of the $t$-$V$-$V'$ model.\cite{Schmitteckert2004,Mishra2011} For
$p=4$, we find representative results by fixing $U_1 = U_3 = 4$, $U_4 = 1$ and
again varying $U_2$. According to the phase diagram in
Fig.~\ref{fig:exampleCDW}, this covers the region occupied by the phases
$(\bullet \circ)$, $(\bullet\bullet\bullet\circ\circ\circ)$ and
$(\bullet\bullet\circ\circ)$ in the atomic limit. These three phases are
remarkably robust against the introduction of a finite kinetic energy, but to a
varying degree, which leads to the unconventional transition scenarios that are
the main result of this work. In contrast, the other two phases listed for $p=4$
in Tab.~\ref{tab:exampleCDW} occupy a smaller part of phase space and are more
susceptible to suppression by a finite kinetic energy. This is illustrated at
the end of this section for the phase $(\bullet \circ\circ\bullet\circ\bullet
\bullet\circ)$, which is quickly replaced by a bond-ordered phase.

\subsection{Phenomenological signatures}\label{sec:nodis_finitet_results}

\begin{figure}[t]
  \includegraphics[width=0.9\columnwidth]{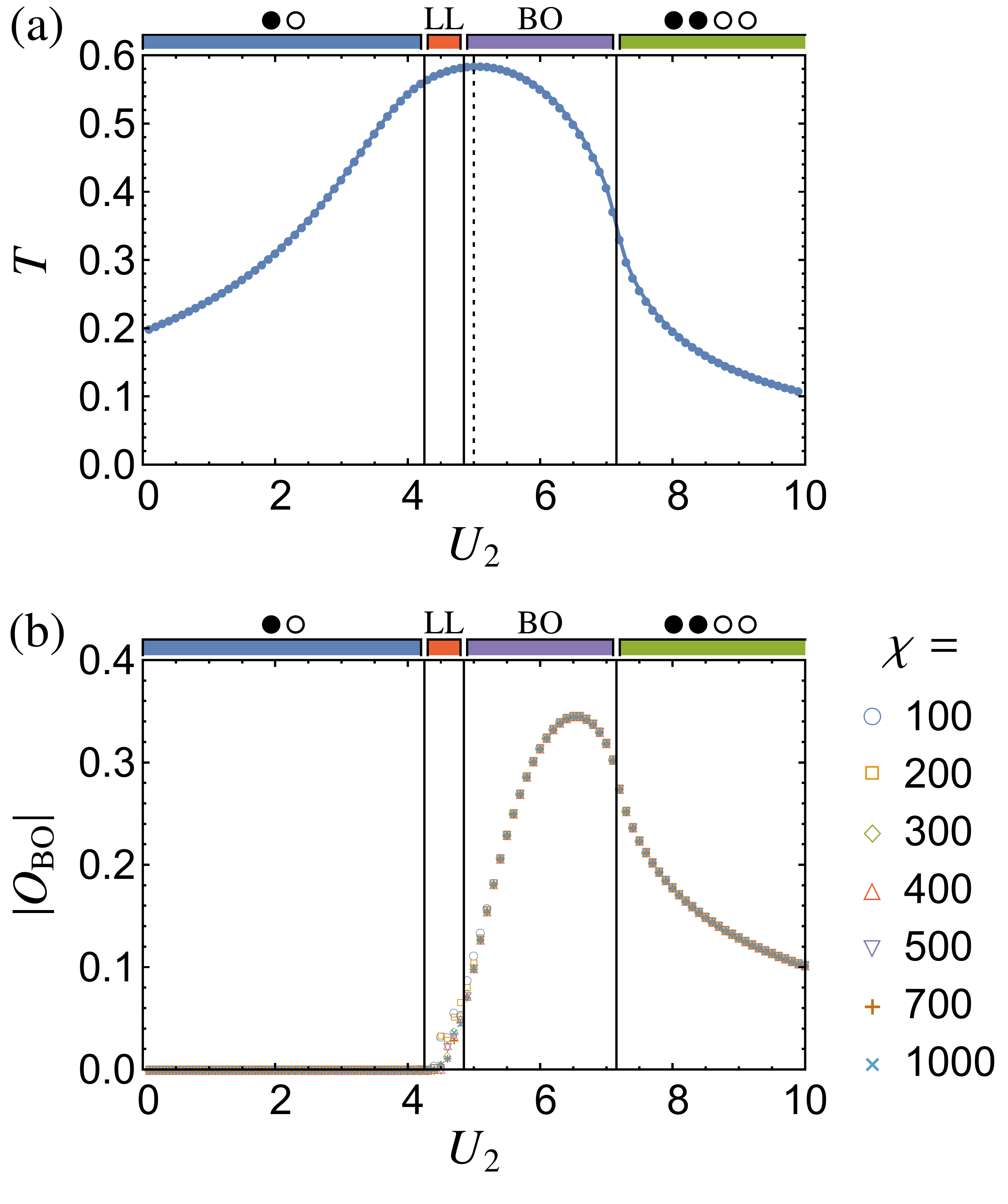}
  \caption{(Color online) Phenomenological signatures of phases and transitions
  in the kinetic energy density $T$ (a) and bond-order parameter
  $O_{\mathrm{BO}}$ (b) at finite kinetic energy parameter $t=1$, for
  interaction range $p=2$ and half filling ($Q=1/2$). In both panels the
  next-nearest-neighbor interaction $U_2$ is varied while the nearest-neighbor
  interaction is set to $U_1=10$. The solid lines indicate the phase transitions
  in this setting, while the dotted line in (a) indicates the transition between
  the two charge-ordered phases in the atomic limit. The kinetic energy density
  only captures a single phase transition  at $U_2 \approx 7.15$, which
  coincides with the transition from the liquid phase into the charge-ordered
  phase $(\bullet \bullet \circ \circ)$. The bond-order parameter vanishes in
  the thermodynamic limit of the phase $(\bullet \circ)$ and the liquid phase,
  but scales differently with increasing bond dimension $\chi$, thereby
  providing signatures of all three phase transitions.
  \label{fig:kinetic_nodisp2}}
\end{figure}

\begin{figure}[t]
  \includegraphics[width=0.9\columnwidth]{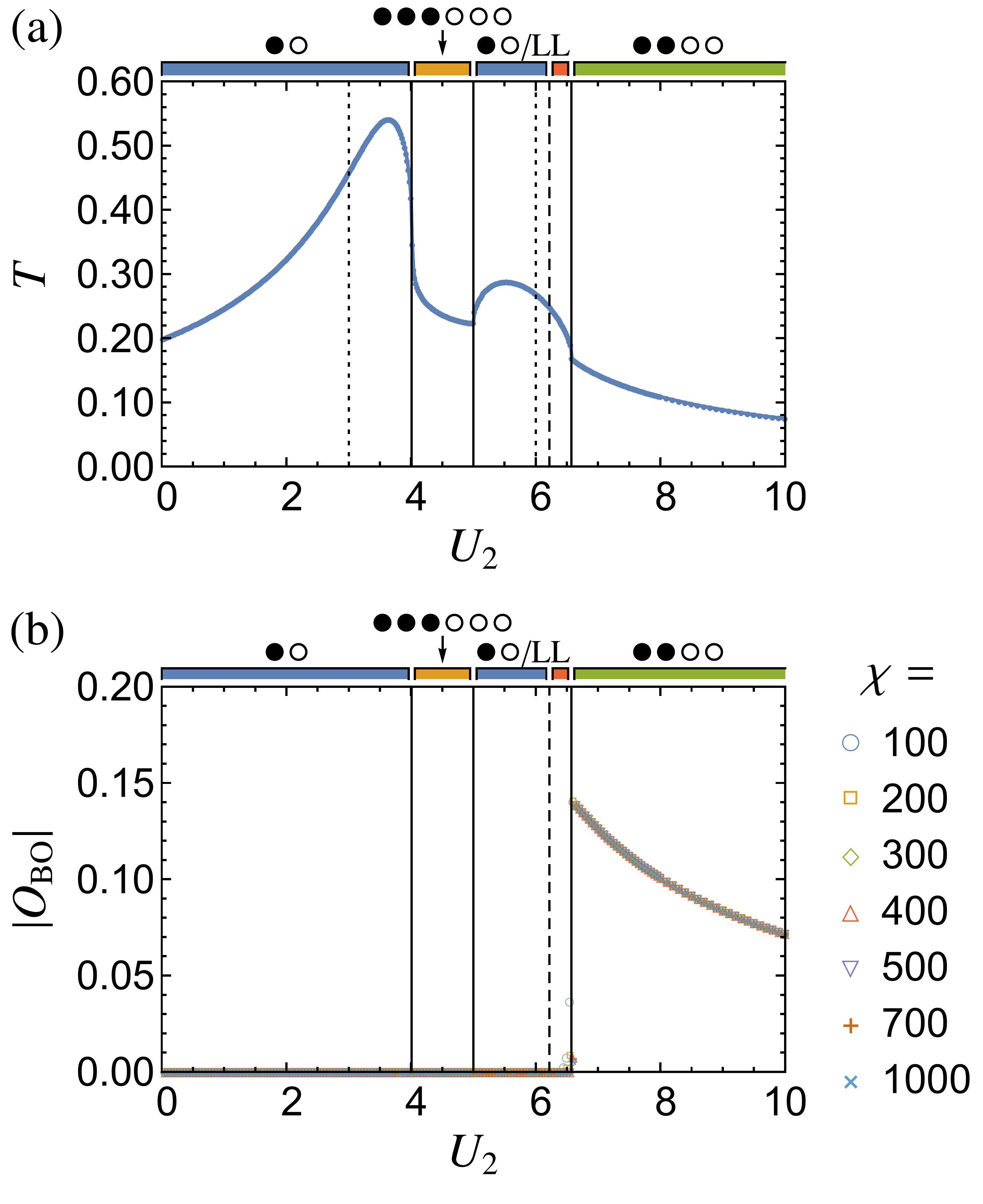}
  \caption{(Color online) kinetic energy density $T$ (a) and bond-order
  parameter $O_{\mathrm{BO}}$ (b) in analogy to Fig.~\ref{fig:kinetic_nodisp2},
  but for increased interaction range $p=4$ and interaction parameters fixed to
  $U_1 = U_3 = 4$ and $U_4 = 1$. In the atomic limit, the transitions between
  the charge-ordered phases occur at $U_2=3,6$ (dotted lines). Discontinuities
  in the derivative of the kinetic energy density indicate three clear phase
  transitions at $U_2 \approx \{ 4.01, 5.00, 6.57\}$ (solid lines). The
  bond-order parameter only captures a single transition into the phase
  $(\bullet\bullet\circ\circ)$, which is the only encountered phase where
  $O_{\mathrm{BO}}$ is finite. Note that the transition between the phases
  $(\bullet\circ)$ and $(\bullet\bullet\circ\circ)$ appears to be abrupt. This
  is verified in the subsequent figures, which also determine the indicated
  nature of the mediating region between the phases
  $(\bullet\bullet\bullet\circ\circ\circ)$ and $(\bullet\bullet\circ\circ)$.
  There, we observe a crossover from a re-emergent charge-ordered state
  $(\bullet\circ)$ to liquid behavior, as indicated by the dashed line.
  \label{fig:kinetic_nodisp4}}
\end{figure}

We first consider the impact of a finite kinetic energy on the phenomenological
level. This is most directly captured by inspection of the kinetic energy
density $T$ [see Eq.~\eqref{eq:t}], which is shown in panels (a) of
Figs.~\ref{fig:kinetic_nodisp2} and ~\ref{fig:kinetic_nodisp4}, and its
staggered version, the bond-order parameter $O_{\mathrm{BO}}$
[Eq.~\eqref{eq:obo}], which is shown in panels (b). We note that the range of
values taken by both quantities is comparable for $p=2$ and $p=4$, which places
us at a similar distance to the atomic limit. The effect of the different
interaction range for both cases is immediately visible.

For $p=2$ (Fig.~\ref{fig:kinetic_nodisp2}), the kinetic energy density increases
as we approach the transition between the two insulating phases in the atomic
limit (dotted line). The analytical behavior of $T$ resolves a single phase
transition, which is signaled by a discontinuity in its first derivative (solid
line). As confirmed by the bond-order parameter, this phase transition coincides
with the transition from the bond-ordered phase to the phase $(\bullet \bullet
\circ\circ)$, where the latter admits a finite values of $O_{\mathrm{BO}}$ as
the cuts $(|\bullet  \bullet | \circ \circ|)$   and $(\bullet | \bullet \circ
|\circ)$ are inequivalent in this phase. The liquid phase is signaled by the
continuing drop of $O_{\mathrm{BO}}$ with increasing bond dimension, as this
order parameter has to vanish in the thermodynamic limit. The bond-order
parameter also vanishes in the phase $(\bullet\circ)$, as the cuts  $(|\bullet
\circ|)$ and  $(\bullet | \circ)$ are equivalent by particle-hole symmetry. The
resulting sequence of phases is marked on top of the panels. The resulting
picture agrees with the previous studies of the $t$-$V$-$V'$ model in
Refs.~\onlinecite{Schmitteckert2004, Mishra2011}, where the phase diagram was
determined from the bond-order parameter and the ground-state curvature. In this
case, therefore, the kinetic energy density carries less detailed information
than the bond-order parameter.

In contrast, for the increased interaction range $p=4$
(Fig.~\ref{fig:kinetic_nodisp4}) we find distinctively more pronounced
signatures of several phases already in the kinetic energy density, with three
clear phase transitions indicated by the analytical behavior of $T$. As we will
confirm below, these coincide with an abrupt transition between the phases
$(\bullet \circ)$ and $(\bullet \bullet \bullet \circ\circ\circ)$, as if still
in the atomic limit; a transition into a mediating region with partially ordered
and partially liquid behavior, to which we will pay particular attention; and
finally the transition into the phase $(\bullet\bullet\circ\circ)$. The
bond-order parameter now carries less insight as it is only finite in the  phase
$(\bullet\bullet\circ\circ)$; in the phase $(\bullet \bullet \bullet
\circ\circ\circ)$ it vanishes as the cuts  $(|\bullet \bullet| \bullet
\circ|\circ\circ|)$  and  $(\bullet |\bullet \bullet |\circ\circ| \circ)$  are
again equivalent by particle-hole symmetry. We therefore do not detect a
separate bond-ordered phase. Instead, as we will argue in the following, the
mediating region contains a crossover between a re-emergent phase $(\bullet
\circ)$  and a liquid phase, resulting in the sequence of phases indicated at
the top of the panels.

\begin{figure}[t]
  \includegraphics[width=0.9\columnwidth]{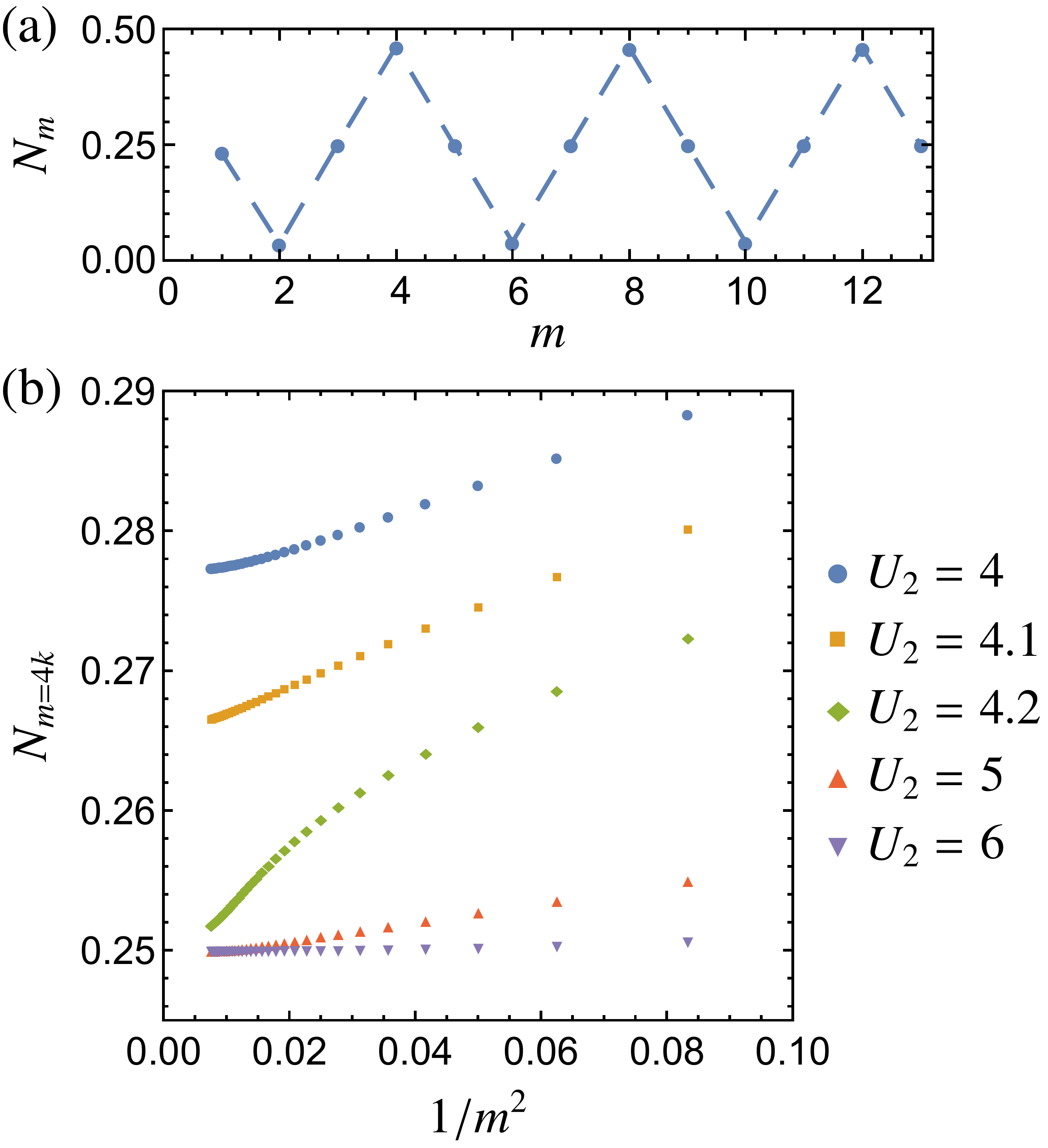}
  \caption{(Color online) (a) Range dependence of the correlation function $N_m$
  for $p=2$, $Q=1/2$, $t=1$, with the interaction potentials $U_1 = 10$ and $U_2
  = 8$ set to values where the system is in the charge-ordered state $(\bullet
  \bullet \circ \circ)$. The function displays a clear oscillatory behavior with
  period 4, reflecting that the charge-ordered character of the phase persists
  at finite kinetic energy. (b) Extrapolation of the correlation function
  $N_{4k}$ with increasing $k$, with $U_1$ and $t$ as above but for various
  values of $U_2$. With increasing range $4k$ this correlation function
  converges to a well-defined value $N_4^\infty$, in accordance to the
  extrapolated correlators stipulated in Eq.~\eqref{eq:nminf}.
  \label{fig:finite_corr_period}}
\end{figure}

\subsection{Correlation functions}

A more detailed characterization of the encountered phases is provided by the
correlation functions $N_m$ defined in Eq.~\eqref{eq:nm}. The utility of these
functions is illustrated in Fig.~\ref{fig:finite_corr_period}(a), where we show
their $m$-dependence for a half-filled system with $p=2$ in the region where we
expect the charge-ordered phase $(\bullet \bullet \circ \circ)$. The correlation
function displays an oscillating behavior in $m$, with a period $P=4$ that
reflects the size of the charge-ordered unit cell. As shown in
Fig.~\ref{fig:finite_corr_period}(b), the limiting quantities  $N_m^{\infty}$
given in Eq.~\eqref{eq:nminf} are indeed well defined.

With these features, the correlation functions give direct insight into the
charge-ordered character of the phases. In the atomic limit, the extrapolated
correlators take values $N_m^{\infty} \in \{ i / P \}$ if the system is in the
insulating phase, where $i$ is an integer. The charge-ordered phase with a unit
cell $(\bullet \circ)$ displays two alternating values $0$ and $1 / 2$, the
phase $(\bullet \bullet \circ \circ)$ admits three possible values alternating
as $(1/4, 0, 1/4, 1/2)$, and the phase $(\bullet \bullet \bullet \circ \circ
\circ)$ admits four possible values alternating as $(1/3,1/6,0,1/6,1/3,1/2)$.
For $t > 0$, the precise values of $N_m^{\infty}$ in a charge-ordered phase are
expected to deviate from the atomic limit, but their periodicity  and the
ordering of the encountered values should be preserved. In the liquid phase, we
expect the long-range correlations to become trivial. The average correlation
function between any two positions in the system should therefore acquire the
value $\langle n_i n_{i + m} \rangle / L^2 \approx \langle n_i \rangle \langle
n_{i + m} \rangle / L^2 = Q^2=1/4$, where we specified the case of half filling.
While these features clearly separate all charge-ordered states, the same value
$1/4$ is also obtained in the bond-ordered phase, which we detected above with
the bond-order parameter.

\begin{figure}[t]
  \includegraphics[width=\columnwidth]{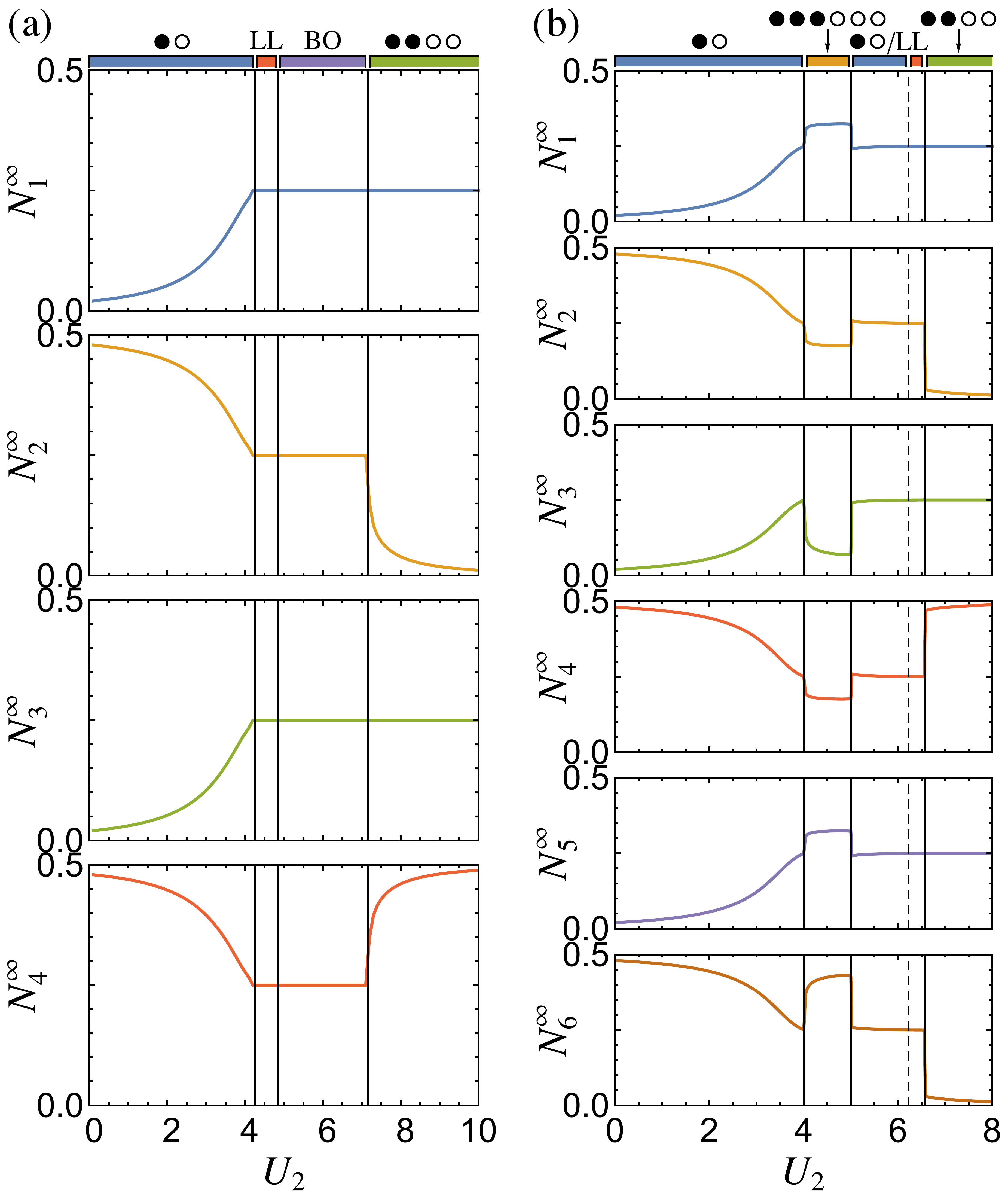}
  \caption{(Color online) Extrapolated correlators $N_m^{\infty}$ for (a) $p =
  2$ and (b) $p = 4$, with system parameters specified as in
  Figs.~\ref{fig:kinetic_nodisp2} and \ref{fig:kinetic_nodisp4}. The phase
  transitions indicated by the solid lines coincide with points where the
  derivative of $N_m^{\infty}$ is discontinuous. The oscillatory behavior of the
  correlators with the range $m$ agrees with the stipulated charge orders. For
  $p=4$, the transition between the phases $(\bullet\circ)$ and
  $(\bullet\bullet\bullet\circ\circ\circ)$ remains abrupt. The behavior in the
  mediating transition region between the phases
  $(\bullet\bullet\bullet\circ\circ\circ)$ and $(\bullet\bullet\circ\circ)$ is
  examined more closely in Fig. \ref{fig:Nminfp4}.
  \label{fig:infinite_correlators}}
\end{figure}

Fig.~\ref{fig:infinite_correlators} shows the extrapolated correlators
$N_m^{\infty}$ for both investigated systems under the same conditions as in
Figs.~\ref{fig:kinetic_nodisp2} and \ref{fig:kinetic_nodisp4}. All
charge-ordered phases can be clearly identified using the expected periodicity
from the atomic limit. We notice that there are regions where all the
correlators reach the value $1/4$, indicating the absence of charge order. The
discontinuities in the derivative of $N_m^{\infty}$ coincide with the phase
transitions detected by $T$ and $O_{\mathrm{BO}}$.

For $p=4$, these results confirm that the transition between the phases
$(\bullet \circ)$ and $(\bullet \bullet \bullet \circ \circ \circ)$ appears to
be sharp, with an undetectable intervening liquid phase. Significantly, the
mediating region between the phases $(\bullet \bullet \bullet \circ \circ
\circ)$ and $(\bullet \bullet \circ \circ)$ indeed exhibits the signatures of
the phase $(\bullet \circ)$. Furthermore, as shown in detail in
Fig.~\ref{fig:Nminfp4}, the values of $N_m^{\infty}$ get progressively closer to
1/4, indicating a possible crossover into the liquid state, as suggested by the
label $(\bullet \circ)$/LL.

\begin{figure}[t]
  \includegraphics[width=0.9\columnwidth]{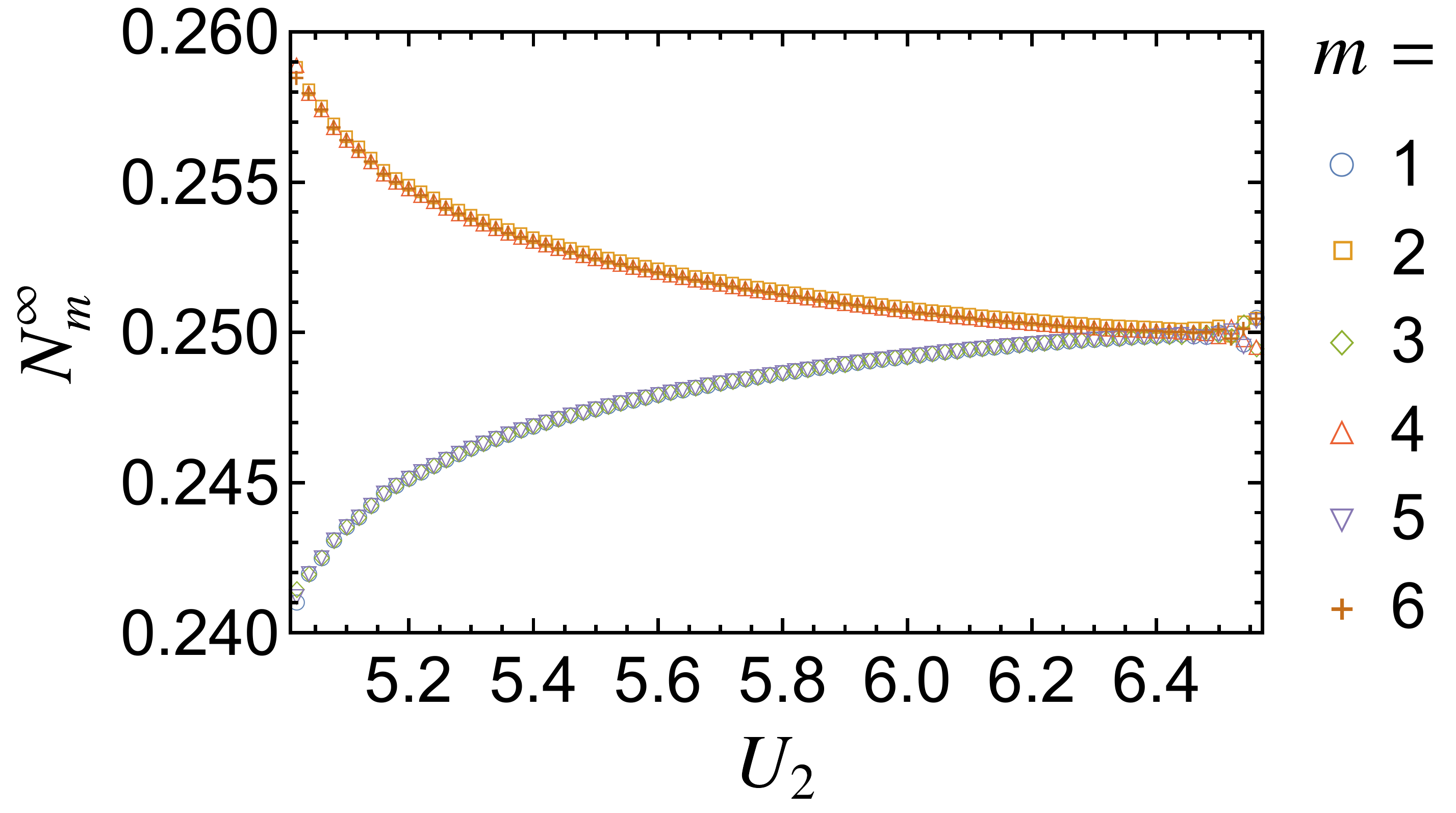}
  \caption{(Color online) Close-up of the extrapolated correlators
  $N_m^{\infty}$ from Fig.~\ref{fig:infinite_correlators} for the mediating
  transition region in the system with $p = 4$. The correlations gradually
  approach the value $1/4$, compatible with a gradual crossover from charge
  order $(\bullet\circ)$ to liquid behavior, where furthermore the results
  converge only slowly.  \label{fig:Nminfp4}}
\end{figure}

\subsection{Entanglement entropy}

To further resolve the details of the mediating transition region we turn to the
entanglement entropy, which is presented in Fig.~\ref{fig:entropy_p24}. In
general, in ordered states the entropy can take multiple values depending on the
position of the cut that bipartites the system. In the case of $p=2$ (panel a),
the mirror symmetry of the insulating phase $(\bullet \circ)$ implies that $S$
remains single-valued, while the insulating phase $(\bullet \bullet \circ
\circ)$ has two possible values, where we account for  particle-hole symmetry
and mirror symmetry. In the bond-ordered phase, the ground state is
characterized by an alternating  local structure, and therefore $S$ again has
two possible values. Finally, in the liquid phase the entropy is single-valued,
but does not converge with increasing bond dimension. Therefore, the entropy
also allows to discriminate the bond-ordered and liquid phases. The transitions
between the different phases are clearly visible in these numerical results, and
agree with the signatures described above.

\begin{figure}[t]
  \includegraphics[width=0.9\columnwidth]{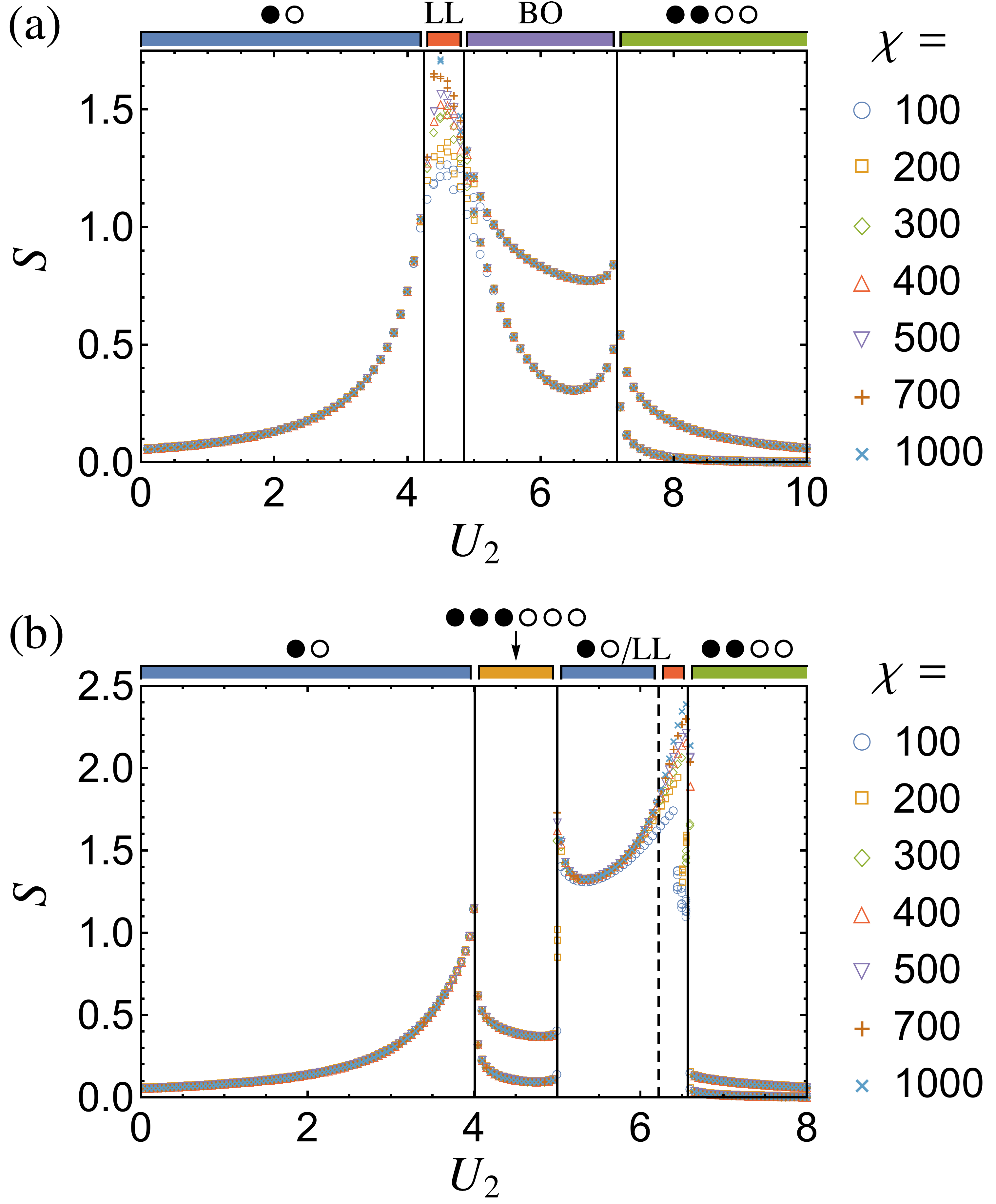}
  \caption{(Color online) Bipartite entanglement entropy for the systems with
  $p = 2$ (a) and $p=4$ (b), with parameters as specified in Figs.\
  \ref{fig:kinetic_nodisp2} and \ref{fig:kinetic_nodisp4}. For $p=4$, no
  critical behavior is detected around the transition between the phases
  $(\bullet\circ)$ and $(\bullet\bullet\bullet\circ\circ\circ)$. However, in the
  transition region between the phases $(\bullet\bullet\bullet\circ\circ\circ)$
  and   $(\bullet\bullet\circ\circ)$ the gradual suppression of charge
  correlations (Fig.~\ref{fig:Nminfp4}) coincides with an onset of finite-size
  scaling with increasing bond-dimension, as expected from a crossover into a
  liquid phase. \label{fig:entropy_p24}}
\end{figure}

Figure~\ref{fig:entropy_p24}(b) shows the entanglement entropy for the
half-filled system with $p = 4$. The entropy can again be multivalued, where the
charge-ordered phases $(\bullet \circ)$ and $(\bullet \bullet \circ \circ)$ and
the liquid phase behave in analogy to the case $p = 2$. In the phase $(\bullet
\bullet \bullet \circ \circ \circ)$ the entropy can have two values, where we
again account for mirror symmetry and particle-hole symmetry. As anticipated
above, the entropy does not detect any indications of a liquid phase between the
phases $(\bullet \circ)$ and $(\bullet \bullet \bullet \circ \circ\circ)$. We
also do not find any indications of the bond-ordered phase, which may be
attributed to the modified energetic conditions from the additional interaction
terms. Most importantly, the results confirm that the transition between the
phases $(\bullet \bullet \bullet \circ \circ \circ)$ and $(\bullet\bullet
\circ\circ)$ is mediated by an ordered state that shares all signatures with the
phase $(\bullet \circ)$. As the ordering in this state approaches the liquid
behavior according to the correlations in Fig.~\ref{fig:Nminfp4}, the system
develops the features of a liquid state, where the entropy continues to increase
with increasing bond dimension. Within the numerically accessible bond
dimensions this takes the form of a transition as indicated above the panel, but
it is also plausible that this behavior indicates a crossover with a rapidly
increasing convergence threshold in the bond dimension.

\subsection{Extended phase diagram}

According to the picture developed above, the phase $(\bullet\circ)$ is present
twice in the phase diagram explored thus far, where it surrounds the  phase
$(\bullet \bullet \bullet \circ \circ \circ)$.  As a phenomenological
explanation for this re-emergent behavior once could suggest that the phase
$(\bullet\circ)$ mediates the transition between the phases  $(\bullet\bullet
\circ\circ)$ and $(\bullet \bullet \bullet \circ \circ \circ)$ as its simpler
structure reflects the required charge-reconfigurations between the latter two
phases. Further insight into the mechanism behind this reemergence is obtained
by sampling the parameter space more widely. For this we keep the potential
energies $U_1 = U_3 = 4, U_4 = 1$ fixed as before and continue to vary  $U_2$,
but also consider values of $t\neq 1$ that supplement the results presented
earlier in this section.

We recall that for $t = 0$, the system is in the atomic limit (see
Sec.~\ref{sec:atomic}), where direct transitions between the charge-ordered
phases occur at $U_2 = 3$ and $U_2 = 6$. For $t\to \infty$, we expect the phase
diagram to display the liquid phase, as the Hamiltonian is then  dominated by
the kinetic term. While we cannot map out the full phase diagram with a high
degree of precision, reasonable estimates of the phases and their transitions
are obtained by limiting the bond dimension $\chi$ to $200$ and $400$ for
representative values of $t$. This leads to the extended phase diagram proposed
in Fig.~\ref{fig:p3phases}. Here, most transitions are captured accurately with
good agreement between the signatures from the extrapolated correlators and the
entropy. The determination of the transition between the phase $(\bullet \circ)$
and the liquid phase requires very high bond dimensions, so that its approximate
position at the accessible bond dimensions is marked by a dashed line.

As can be seen from the diagram, the phase  $(\bullet \bullet \bullet \circ
\circ \circ)$ only exists up to moderate values of the kinetic energy parameter
$t$, while the phase $(\bullet  \circ)$ is distinctively more robust, wraps
around the other phase, and indeed reemerges in the mediating transition region
along the $t=1$ line. At first sight, one would expect that this re-emergent
behavior cannot extend all the way to the atomic limit, as the fully
charge-ordered state $(\bullet \circ)$ acquires a larger energy than the other
charge-ordered states. However, a consistent scenario would see the mediating
state to gradually lose its clear charge order and cross over into the liquid
phase, in analogy to the behavior witnessed in the correlations of Fig.\
\ref{fig:Nminfp4} and the entropy in Fig.\ \ref{fig:entropy_p24}. This assertion
is difficult to verify, as the mediating phase becomes confined to a very small
part of phase space as one approaches the atomic limit. This complication does
not apply to the region $t\approx 1$, where the $(\bullet \circ)$ phase clearly
wraps around the $(\bullet \bullet \bullet \circ \circ \circ)$  phase, resulting
in its re-emergent behavior.

\begin{figure}[t]
  \includegraphics[width=0.8\columnwidth]{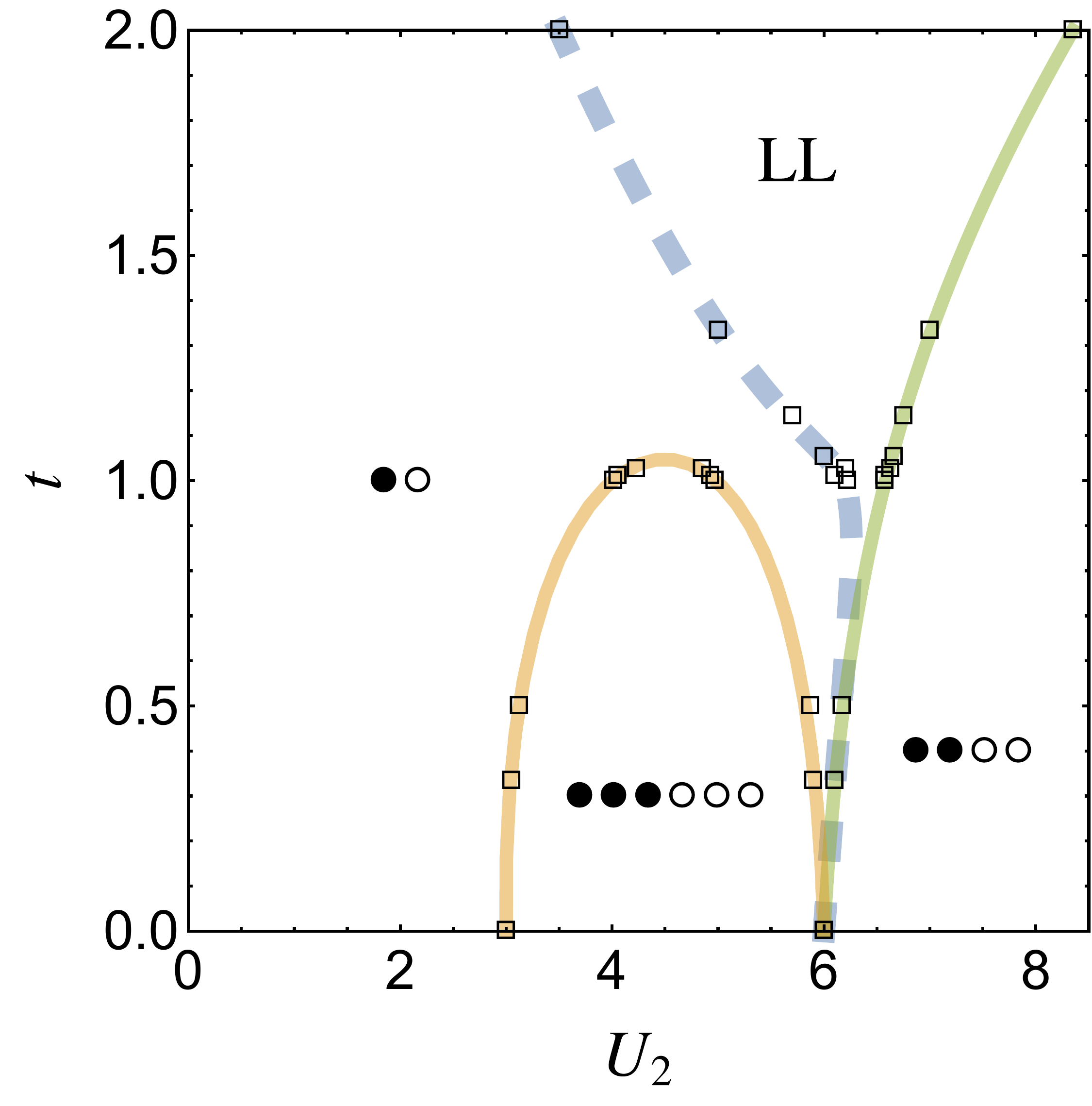}
  \caption{(Color online) Proposed phase diagram for the model from
  Eq.~\eqref{eq:Ham} with interaction range $p = 4$, with interaction parameters
  $U_1 = U_3 = 4, U_4 = 1$ while $U_2$ and the kinetic energy parameter $t$ are
  varied. The phase transitions are determined from the extrapolated correlators
  and the entanglement entropy, in analogy to
  Figs.~\ref{fig:infinite_correlators} and \ref{fig:entropy_p24}, which
  correspond to the case $t=1$. The dashed line indicates the detected onset of
  finite-size scaling with the bond dimension in the crossover from the phase
  $(\bullet \circ)$ to the liquid phase. \label{fig:p3phases}}
\end{figure}

\begin{figure}[t]
  \includegraphics[width=\columnwidth]{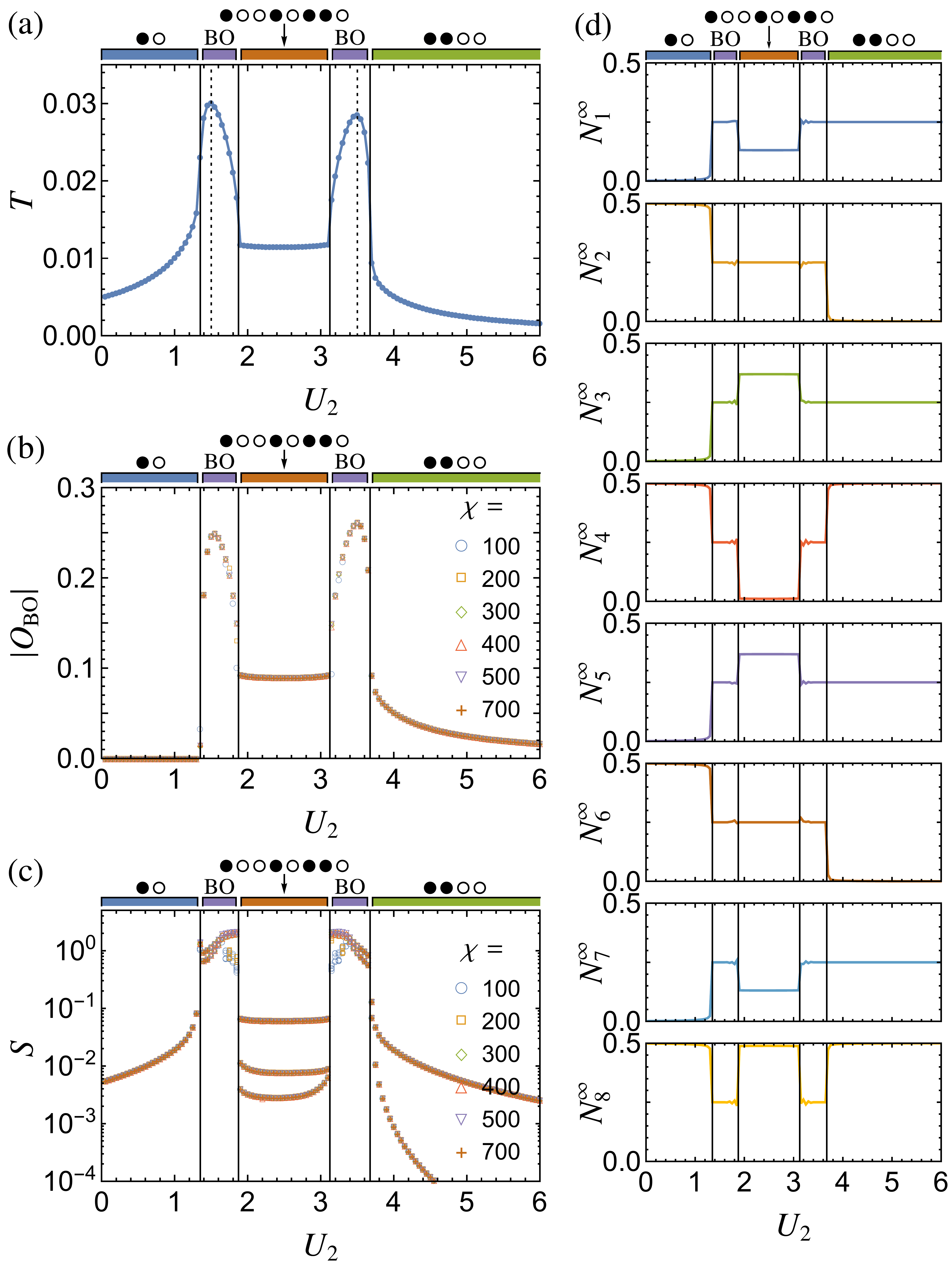}
  \caption{(Color online) Effect of a small kinetic energy parameter $t=0.1$ on
  the phase $(\bullet\circ\circ\bullet\circ\bullet\bullet\circ)$ in the
  half-filled system with $p=4$, as captured by the kinetic energy density $T$
  (a), the bond-order parameter $O_{\mathrm{BO}}$ (b), the bipartite
  entanglement entropy $S$ (c), and the extrapolated correlators $N_m^\infty$
  (d). In contrast to the case described in Fig.~\ref{fig:kinetic_nodisp4},  we
  now set $U_1=4$ and $ U_3= U_4 = 1$, while varying $U_2$ as before. Already at
  the chosen small kinetic energy parameter, the phase
  $(\bullet\circ\circ\bullet\circ\bullet\bullet\circ)$ is driven out by a
  bond-ordered phase, which intervenes in the transition to the other
  charge-ordered states covered in this parameter range. \label{fig:phase5}}
\end{figure}

\subsection{Fragile phases and bond order}\label{sec:fragile}

According to Table~\ref{tab:exampleCDW}, for $p=4$ the considerations above
cover three of the five possible charge-ordered phases identified in the atomic
limit. As the phase $(\bullet\bullet\bullet\bullet\circ\circ\circ\circ)$ is
confined to a small part of parameter space already in the atomic limit (see
Fig.~\ref{fig:exampleCDW}), we here illustrate the susceptibility to a finite
kinetic energy for the phase
$(\bullet\circ\circ\bullet\circ\bullet\bullet\circ)$, which displays the most
complex charge order. To explore this phase we set $U_1=4, U_3= U_4 = 1$ and
again vary $U_2$. We then find that the phase is absent at $t=1$, but can be
detected if we significantly reduce the kinetic energy parameter to $t=0.1$. The
corresponding results are shown in Fig.~\ref{fig:phase5}. We find clear
signatures of three charge-ordered phases, which occur in the sequence
$(\bullet\circ)$, $(\bullet\circ\circ\bullet\circ\bullet\bullet\circ)$,
$(\bullet\bullet\circ\circ)$ in agreement with the atomic limit. However, the
support of the phase $(\bullet\circ\circ\bullet\circ\bullet\bullet\circ)$ is
already much reduced. This occurs in favor of two surrounding regions that both
support a purely bond-ordered phase with no residual charge order, and thereby
share the same characteristics as the bond-ordered phase encountered for $p=2$.
Note that the phase $(\bullet\circ\circ\bullet\circ\bullet\bullet\circ)$ also
admits a finite bond-order parameter, as does again the phase
$(\bullet\bullet\circ\circ)$ already discussed above. These results not only
demonstrate the susceptibility of the phase
$(\bullet\circ\circ\bullet\circ\bullet\bullet\circ)$ to suppression by a finite
kinetic energy, but also show that bond-ordered phases can still occur at this
increased interaction range.

\section{Effects of disorder}\label{sec:disorder}

We now determine the effects of disorder, which is introduced into the
Hamiltonian according to Eq.~\eqref{eq:Hdis}. To study these effects numerically
within the adopted framework we choose disorder configurations that remain
compatible with the previously encountered charge-ordered states, so that these
do not experience any artificial frustration. This requires a disordered unit
cell of size $L$ that is commensurable with all the possible insulating phases
present in the atomic limit, hence a multiple of 4 in the half-filled system
with $p = 2$ and a multiple of 24 in the half-filled system with $p=4$. By
inspecting the variations of our results for $T$, $O_{\mathrm{BO}}$ and $S$ with
$L$ for moderate to strong values of the disorder, we have found it
sufficient to set $L = 20$ for $p = 2$ and keep $L = 24$ for $p = 4$, which has
the added benefit of retaining nontrivial extrapolated correlators $N_m^\infty$
as discussed below. The limit of very weak disorder would require an
ever-increasing disordered unit cell that keeps up with the increasing
localization length, which is beyond the practical scope of the adopted
iDMRG/iMPS approach.

In the atomic limit $t=0$, disorder encourages the fragmentization of
charge-ordered states as the energy expense of a charge configuration can be
overcompensated by the energetic gain from the on-site potential. Furthermore,
previously degenerate configurations such as $(\bullet \bullet \circ \circ)$,
$(\circ \bullet \bullet \circ)$, $(\circ \circ \bullet \bullet)$ and $(\bullet
\circ \circ \bullet)$ now acquire different energies. These reconfigurations
have a direct effect on the long-range correlations, which can be expected to
persist also at finite kinetic energy $t\neq 0$. In this general case, we would
expect the disorder also to  localize the liquid phase, so that the phase space
regions with critical behavior should be suppressed.

\begin{figure}[t]
  \includegraphics[width=0.95\columnwidth]{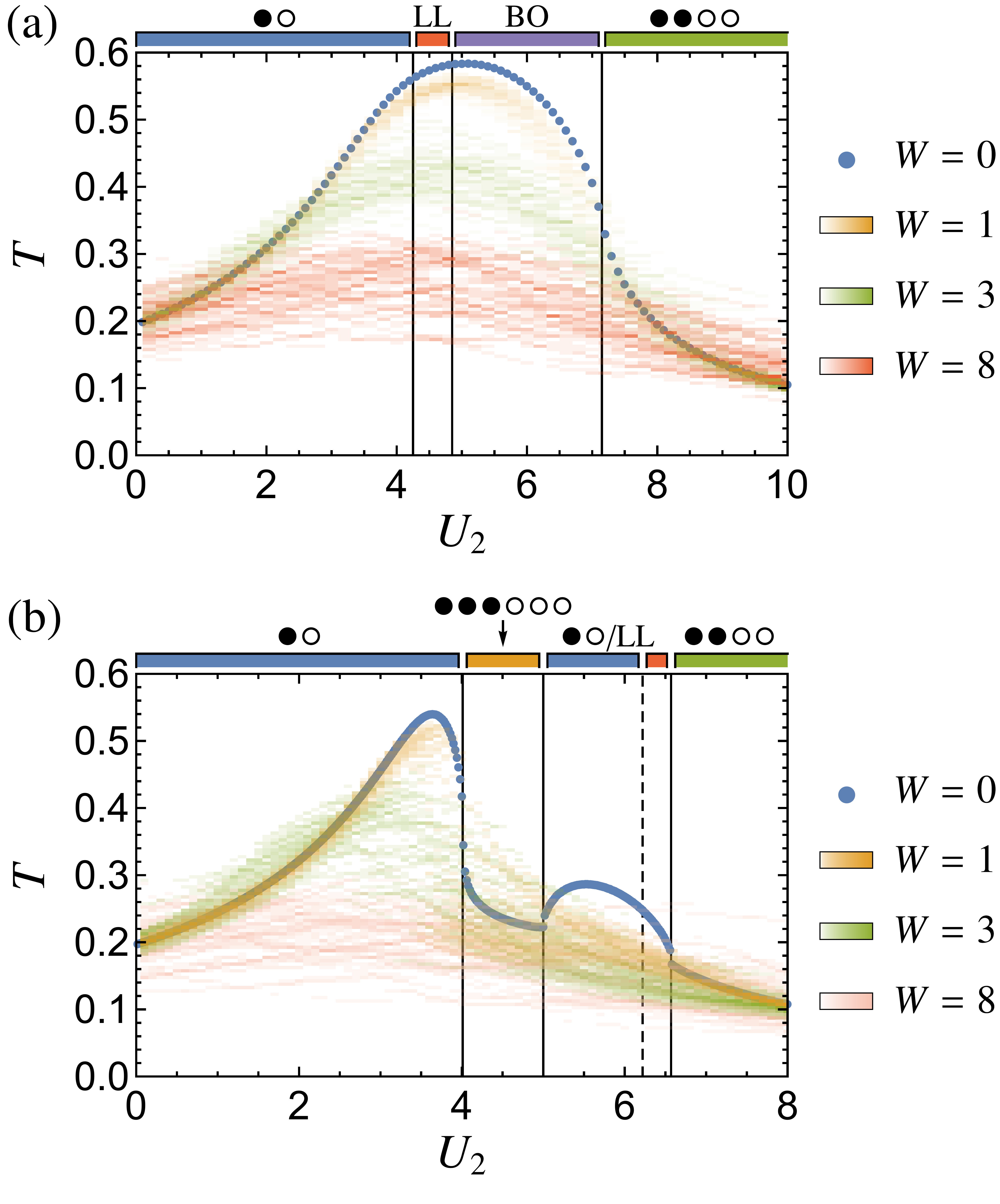}
  \caption{(Color online) Effect of disorder on the kinetic energy density for
  $p = 2$ (a) and $p=4$ (b). The disorder-free case corresponds to
  Figs.~\ref{fig:kinetic_nodisp2} and \ref{fig:kinetic_nodisp4}, and is included
  for guidance. For each finite value of the disorder, the data represent a
  density plot accumulated over 100 disorder realizations. The size of the
  disordered unit cell is $L=20$ ($p=2$) and $L=24$
  ($p=4$).\label{fig:kinetic_dis_change_W_many_p24}}
\end{figure}

On the phenomenological level these anticipated tendencies are again well
captured by the kinetic energy density $T$, as shown by the disorder-averaged
density plots in Fig.~\ref{fig:kinetic_dis_change_W_many_p24}. For $p=2$ (panel
a) the density $T$ drops significantly for increasing disorder strength $W$, in
particular in the transition region between the charge-ordered phases. At the
same time the spread of values of $T$ increases, and all the features present on
the $W = 0$ plot become progressively washed out  so that at $W = 8$ the kinetic
energy becomes essentially independent of $U_2$. A similar trend is present for
$p = 4$ (panel b), where the two prominent peaks present for $W = 0$ get washed
out as one increases the disorder strength. These results are consistent with
the formation of a universal fragmented state at large disorder strength.

\begin{figure}[t]
  \includegraphics[width=\columnwidth]{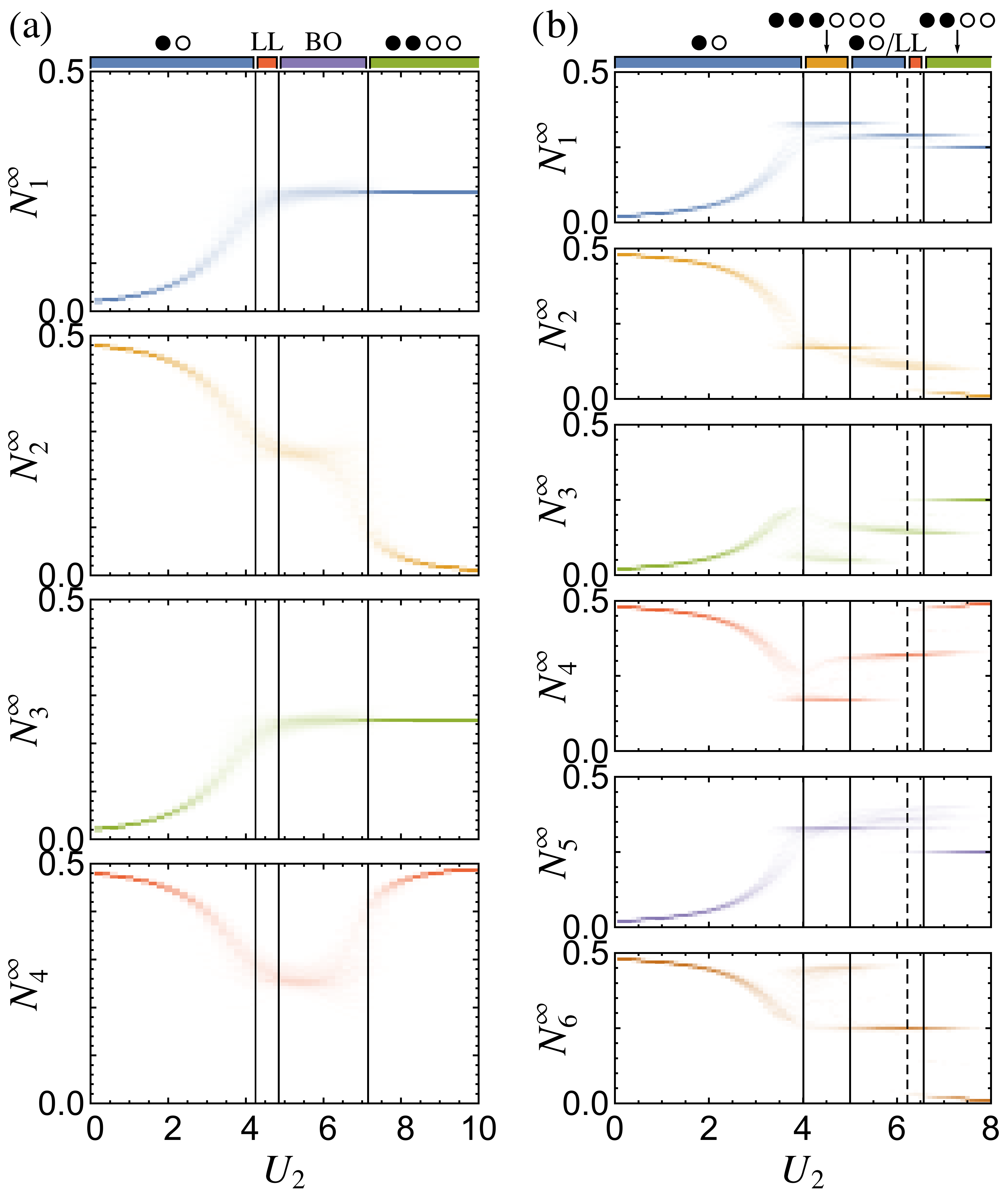}
  \caption{(Color online) Disorder-averaged density plots of the extrapolated
  correlators $N_m^{\infty}$ for the disordered systems specified in
  Fig.~\ref{fig:kinetic_dis_change_W_many_p24}, evaluated at disorder strength
  $W=1$. These correlators remain nontrivial because of the finite size of the
  disordered unit cell. \label{fig:correlators_dis_W1}}
\end{figure}

\begin{figure}[t]
  \includegraphics[width=\columnwidth]{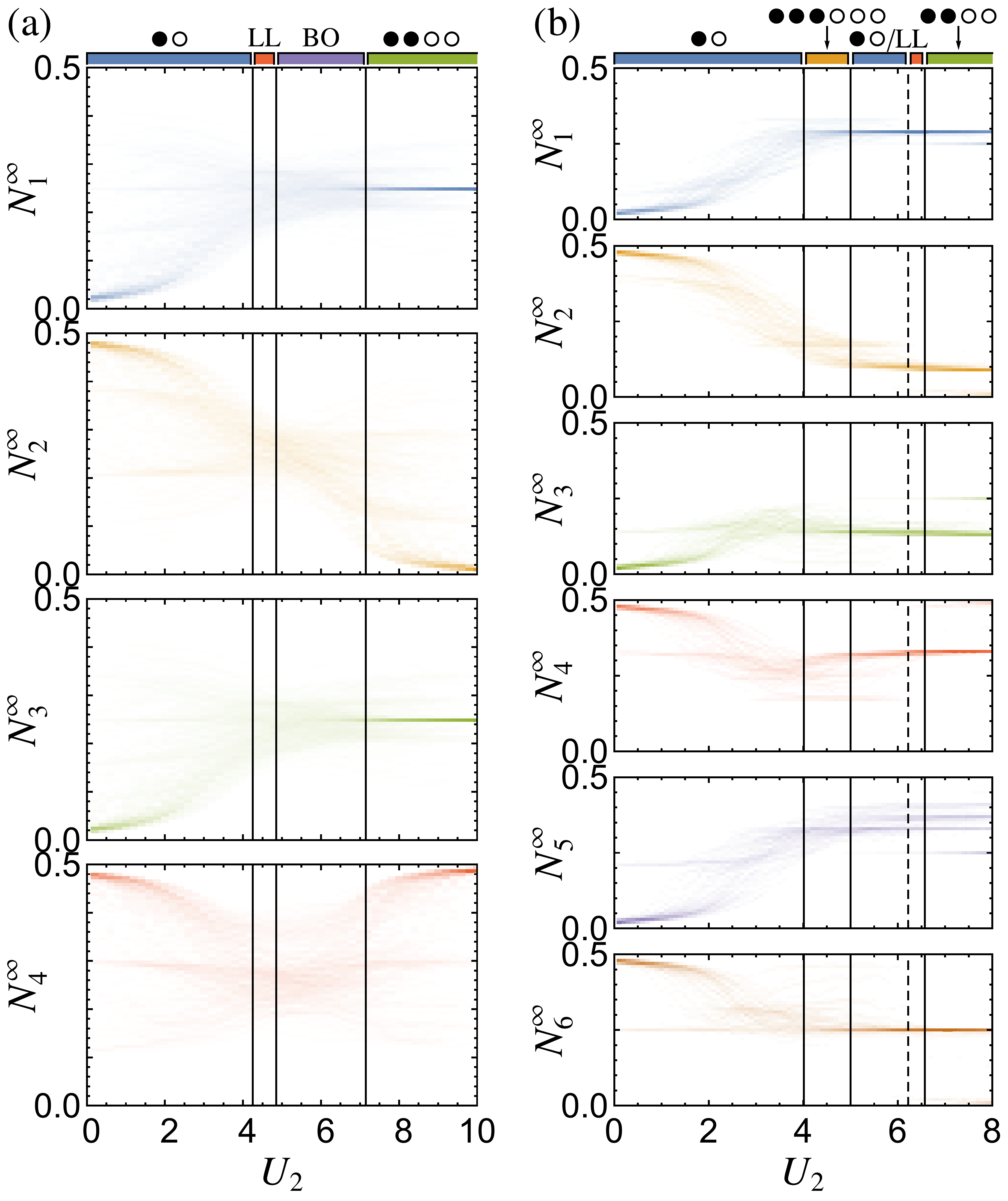}
  \caption{(Color online) Disorder-averaged density plots of the extrapolated
  correlators $N_m^{\infty}$ as in Fig.~\ref{fig:correlators_dis_W1}, but for
  disorder strength $W=3$. \label{fig:correlators_dis_W3}}
\end{figure}

\begin{figure}[t]
  \includegraphics[width=\columnwidth]{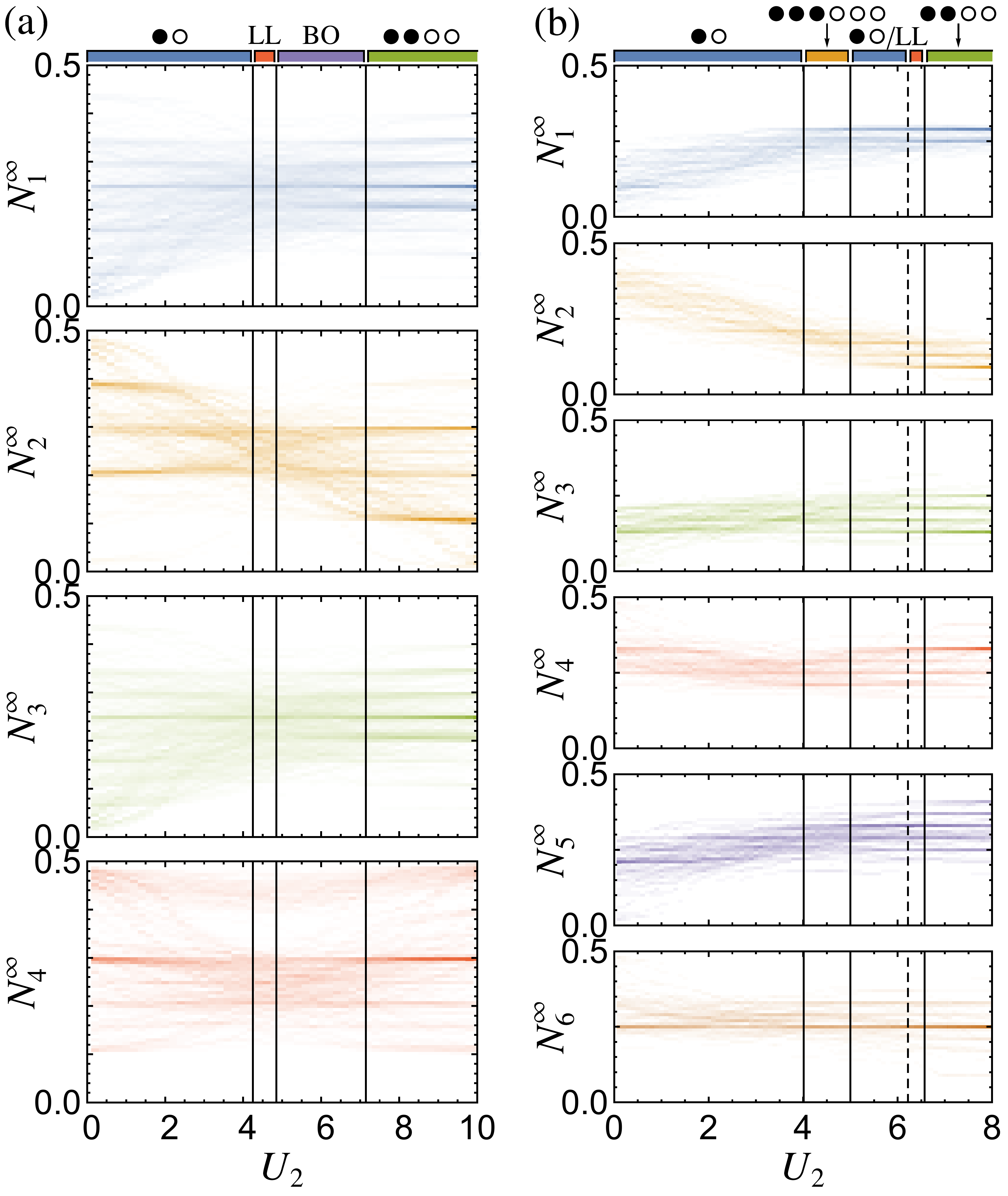}
  \caption{(Color online) Disorder-averaged density plots of the extrapolated
  correlators $N_m^{\infty}$ as in Figs.~\ref{fig:correlators_dis_W1} and
  \ref{fig:correlators_dis_W3}, but for disorder strength
  $W=8$.\label{fig:correlators_dis_W8}}
\end{figure}

Insight into the gradual formation of such a state is given by the correlation
functions. Figures\ \ref{fig:correlators_dis_W1}, \ref{fig:correlators_dis_W3},
and \ref{fig:correlators_dis_W8} show disorder-averaged density plots of the
extrapolated correlators $N_m^\infty$ for disordered systems with $W=1$, $W=3$
and $W=8$, respectively. Note that these correlators are expected to be trivial
(equaling 1/4) at any finite disorder strength in the thermodynamic limit
$L\to\infty$, but here retain a nontrivial structure as $L$ is finite. For a
small disorder strength ($W = 1$, Fig.\ \ref{fig:correlators_dis_W1}), the
correlators for $p=2$ behave very similarly to the non-disordered case. On the
other hand, for $p = 4$ the same disorder strength already has a distinct effect
on the mediating transition region between the phases
$(\bullet\bullet\bullet\circ\circ\circ)$ and $(\bullet\bullet\circ\circ)$, where
the re-emergent phase $(\bullet\circ)$  and the liquid phase are quickly
replaced in favor of a disordered insulating phase with a nonuniform charge
structure. As we increase the disorder ($W = 3$, Fig.\
\ref{fig:correlators_dis_W3}), the correlators $N_m^\infty$ develop distinct
ridges close to rational values $i/L$, with $i=0,\ldots,L/2$ (emphasizing again
the role of the finite disordered unit cell). These ridges are most prominent in
the ranges formerly occupied by the charge-ordered states, where this behavior
is consistent with their fragmentization. In the former transition regions the
correlators cover a broad and continuous range of values. For strong disorder
($W = 8$, Fig.\ \ref{fig:correlators_dis_W8}), however, the rational ridges
spread out over the whole parameter range, which is consistent with the
emergence of a universal fragmented insulating state. Interestingly, this state
still appears to carry some characteristic ordering features. For example, in
the system with $p = 2$, the correlators $N_2^\infty$ and $N_4^\infty$  prefer
rational values $L/i$ with even $i$.

\begin{figure}[t]
  \includegraphics[width=0.95\columnwidth]{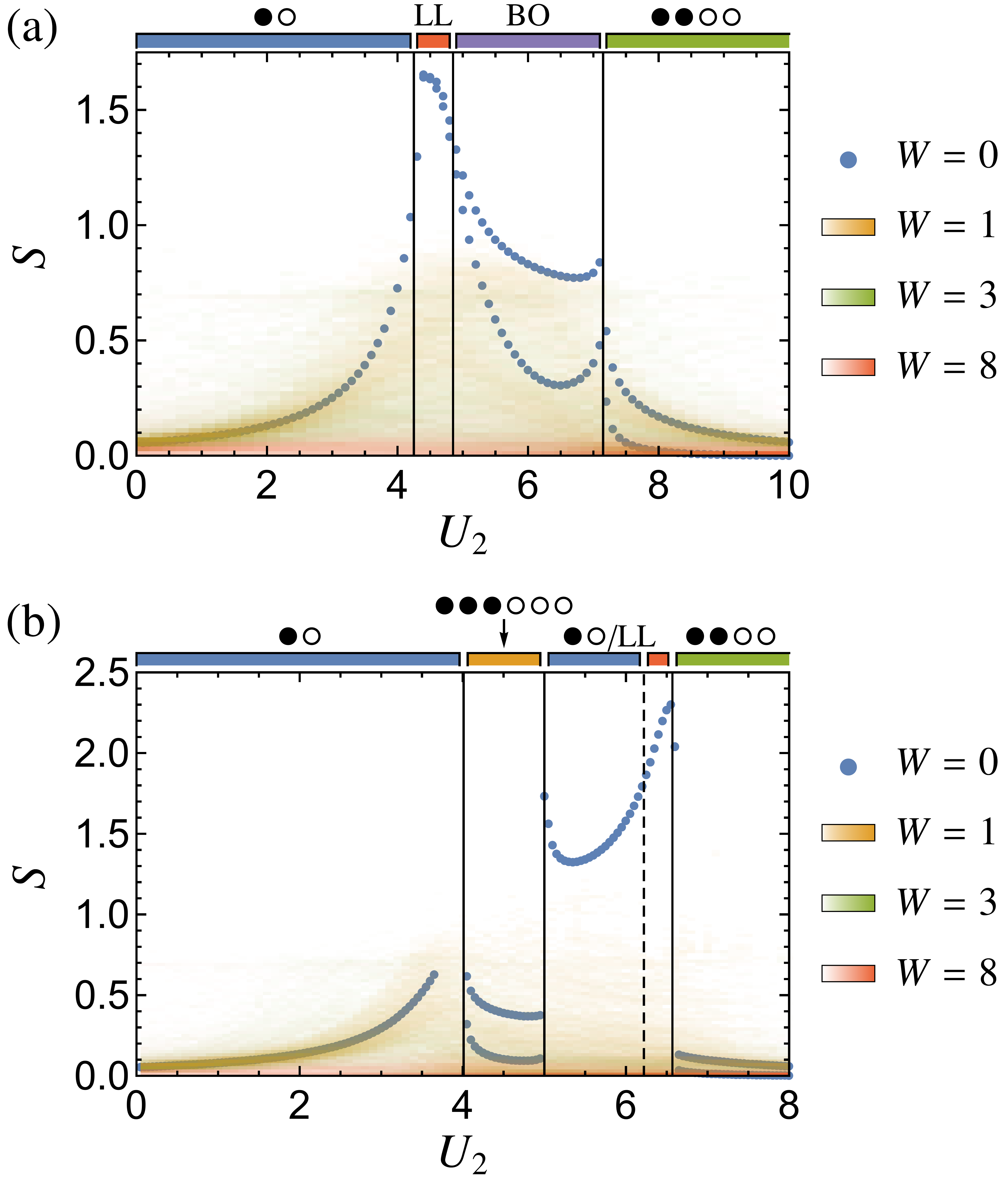}
  \caption{(Color online) Disorder-averaged density plots of the entanglement
  entropy $S$ for the disordered systems specified in
  Fig.~\ref{fig:kinetic_dis_change_W_many_p24}.\label{fig:entropy_dis}}
\end{figure}

Finally, as shown in Fig.~\ref{fig:entropy_dis} the disorder also has a
significant effect on the entanglement entropy. In particular, the entropy
develops strong sample-to-sample fluctuations already for small disorder
strength, and its value in the previously liquid phase is already noticeably
reduced. For large disorder strength the entanglement entropy becomes very small
across the whole parameter range.

\section{Conclusions}\label{sec:conclusions}

In summary, we have investigated the interplay of charge-ordered fermionic
insulating phases that arise from the competition of finite-range interactions
in one dimension, based on the fermionic lattice model in Eq.\eqref{eq:Ham} and
a range of phenomenological and fundamental quantities described in
Sec.~\ref{sec:model}. In the atomic limit of a vanishing kinetic energy term, we
observe a proliferation of competing phases, which for large interaction range
can display remarkably rich internal ordering (see Sec.~\ref{sec:atomic} and the
Appendix). For a finite kinetic energy (Sec.~\ref{sec:finitet}), significant
differences are found already for moderate ranges  $p$ of the interactions, as
we explored by comparing the cases  $p=2$ (for which we recover the known
phenomenology of phase transitions mediated by liquid and bond-ordered phases)
and $p=4$. While in the latter case some complex charge-ordered phases are
quickly suppressed by a finite kinetic energy (see e.g. Fig.~\ref{fig:phase5}),
we observe that the increased variety of competing phases with increasing
interaction range does not imply an essential loss of insulating properties of
the system; see the corresponding panels and phase diagrams in
Figs.~\ref{fig:kinetic_nodisp4}, \ref{fig:infinite_correlators}, and
\ref{fig:entropy_p24}. Instead, we observe, as our two main results, the
survival of apparently direct transitions between two charge-ordered phases
mimicking the atomic limit, as well as the appearance of mediating phases that
display re-emergent simple charge order and exhibit a crossover to liquid
behavior at one of the phase boundaries (see in particular
Fig.~\ref{fig:Nminfp4}).
These two transition scenarios supplement the liquid and bond-ordered phases
encountered in previous studies with a small interaction range, leading to a
rich variety of phases, transitions and crossovers in the system (see
Fig.~\ref{fig:p3phases}). Disorder (explored in Sec.~\ref{sec:disorder}) has the
expected effect of gradual fragmentization and localization of the insulating
and liquid phases, which is particularly visible in the density-density
correlation functions of a finite disordered unit cell
(Figs.~\ref{fig:correlators_dis_W1}-\ref{fig:correlators_dis_W8}).

The results in this work are based on an analytical classification of
charged-ordered states in the atomic limit and extensive numerical
investigations at a moderate finite kinetic energy, both applied to a canonical
model of interacting spinless fermions on a discrete one-dimensional chain.
Complementary approaches could provide useful insights into the exact nature of
the observed transitions and crossovers. Numerically, this could be achieved by
exploring the limit of a very large or small kinetic energy, in which the
approach adopted in this work scales less favorably, or by the investigation of
alternative models, such a spin chains or spinful fermions. Further analytical
progress could be made perturbatively close to the atomic or free limit, or
phenomenologically by field-theoretical approaches of effective, possibly
continuous counterparts of the studied system. In general, our work should
motivate efforts to identify and classify the possible transition scenarios in
systems where the kinetic energy competes with several interactions of different
range. This competition should also persist for excited states, including for
disordered systems that may display many-body localization.
\cite{Nandkishore2015, Abanin2017} These endeavors are left for future
considerations.

\begin{acknowledgments}
We gratefully acknowledge useful discussions with Neil Drummond, Jens Bardarson,
and Fabian Heidrich\mbox{-}Meisner. This research was funded by EPSRC via Grant
No.~EP/P010180/1. Computer time was provided by Lancaster University's High-End
Computing facility. All relevant data present in this publication can be
accessed at \url{http://dx.doi.org/10.17635/lancaster/researchdata/234}.
\end{acknowledgments}

\appendix

\section{Identification of charge-ordered insulating phases in the atomic limit}\label{sec:app1}

In this Appendix we provide classifications of charge-ordered phases at
principal critical densities $Q=Q_m=1/m$ for interaction ranges $p\leq 6$. We
start with the instructive case of $Q=Q_p=1/p$, where the classification can be
carried out for all $p$.

\subsection{Critical density $Q = 1 / p$}\label{sec:insulating-construction}

Assume for the moment that $U_p\ll U_m$ so that the preferable distance between
two fermions is $p$. We then  say that $U_p$ orders the fermions in the ground
state. For example, a charge sequence $\bullet \underbrace{\circ \circ \cdots
\circ}_{p - 1} \bullet \underbrace{\circ \circ \cdots \circ}_{p - 1} \bullet$
has a lower energy than the sequence $\bullet \underbrace{\circ \circ \circ
\cdots \circ}_p \bullet \underbrace{\circ \cdots \circ}_{p - 2} \bullet$. The
ground state has the simple form
\begin{equation}
  \bullet \underbrace{\circ \circ \cdots \circ}_{p - 1} \bullet
  \underbrace{\circ \circ \cdots \circ}_{p - 1} \bullet \underbrace{\circ
  \circ \cdots \circ}_{p - 1} \cdots,
\end{equation}
and its energy is $E_1 = (L/p) U_p = N U_p$, where $L$ is the considered system
size and $N=L/p$ the number of particles.

Next let us inspect how a low value of $U_{p - 1}$ can order the fermions. This
cannot be based on a repeating sequence of segments $\bullet \underbrace{\circ
\circ \cdots \circ}_{p - 2}$, as this would not result in the correct density $1
/ p$. However, by addition of segments $\bullet \underbrace{\circ \circ \cdots
\circ}_p$ we can tailor the density without changing the energy of the system. A
representative corresponding ground-state configuration is
\begin{equation}
  \begin{array}{|l|}
  \hline
  \bullet \underbrace{\circ \circ \cdots \circ}_{p - 2} \bullet
  \underbrace{\circ \circ \cdots \circ}_p\\
  \hline
  \end{array} \begin{array}{|l|}
  \hline
  \bullet \underbrace{\circ \circ \cdots \circ}_{p - 2} \bullet
  \underbrace{\circ \circ \cdots \circ}_p\\
  \hline
  \end{array} \cdots,
\end{equation}
which gives us the correct density $Q = 1 / p$, and results in the energy $E_2 =
(L / 2 p) U_{p - 1} = (N / 2) U_{p - 1}$. Note that this ground state is highly
degenerate---the sections of $p$ and $p-2$ unoccupied sites can be freely
arranged along the system; e.g., all sections with $p - 2$ unoccupied sites
could be placed besides each other without changing the energy of the system.

If one follows this prescription for the general case of ordering driven by
$U_{p-n}$, one obtains ground states of the representative structure
\begin{equation}
  \begin{array}{|l|}
  \hline
  \bullet \underbrace{\circ \circ \cdots \circ}_{p - n} \overbrace{\bullet
    \underbrace{\circ \circ \cdots \circ}_p \bullet \underbrace{\circ \circ
      \cdots \circ}_p \bullet \underbrace{\circ \circ \cdots \circ}_p \cdots}^{n
    - 1 \text{ times}}\\
  \hline
  \end{array} \cdots,
\end{equation}
which have an energy
\begin{eqnarray}
  E_n & = & \frac{L}{1 + p - n + (n - 1) (p + 1)} U_{p + 1 - n} \\
  & = & \frac{L}{n p} U_{p - n + 1} = \frac{N}{n} U_{p - n + 1} . \nonumber
\end{eqnarray}
Again, these ground states are highly degenerate.

We can now determine the conditions in which an arbitrary phase (designated by
step $n$) will dominate the charge ordering. This requires
\begin{equation}
  \underset{k \neq n}{\bigforall} E_n < E_k \quad \Rightarrow \quad
  \underset{k \neq n}{\bigforall} U_{p - n + 1} < \frac{n}{k} U_{p - k +
    1} .
\end{equation}
Renaming $\alpha = p - n + 1$ and $\beta = p - k + 1$, we arrive at the
condition
\begin{equation}
  \underset{\beta \neq \alpha}{\bigforall} U_{\alpha} < \frac{p - \alpha +
  1}{p - \beta + 1} U_{\beta} ,
\end{equation}
which in the main text is expressed as Eq.~\eqref{eq:udom}.

If this condition is fulfilled then the phase with energy $E_{\alpha} = [N/(p -
\alpha + 1)] U_{\alpha}$ is dominant and the ground state consists of $N/(p -
\alpha + 1)$ segments $\bullet \underbrace{\circ \circ \cdots \circ}_{\alpha -
1}$ and $N (p - \alpha)/(p - \alpha + 1)$ segments $\bullet \underbrace{\circ
\circ \cdots \circ}_p$. The ground-state degeneracy is given by
\begin{equation}
  f = \begin{cases}
  \begin{pmatrix}
  N\\
  N / (p - \alpha + 1)
  \end{pmatrix}
  \cdot p & \text{if } 2 \alpha > p\\
  \begin{pmatrix}
  N \frac{p - \alpha}{p - \alpha + 1}\\
  N / (p - \alpha + 1)
  \end{pmatrix}
  \cdot \frac{p (p - \alpha + 1)}{p - \alpha} & \text{otherwise.}
  \end{cases}
\end{equation}
For $2 \alpha \leqslant p$, the degeneracy count reflects the requirement to
exclude cases where blocks of structure $\bullet \underbrace{\circ \circ \cdots
\circ}_{\alpha - 1}$ are adjacent, which then increases their energy by $U_{2
\alpha}$.

\subsection{General properties at higher critical
densities}\label{sec:properties}

To construct the charge-ordered phases at larger critical densities $Q_m=1/m$
with $m<p$, we rely on the following two general properties:

\begin{theorem}
  In any atomic charge configuration of density $Q$, there is at least one
  sequence of  $1 / Q - 1$ or more unoccupied sites.\label{1-CDWth1}
\end{theorem}

\begin{proof}
  When the particles are evenly spread out over the system they are $1 / Q$
  sites apart, \textit{i.e.}, separated by $1/Q-1$ unoccupied sites,
  corresponding to the configuration
  \begin{equation}
  \bullet \underbrace{\circ \circ \cdots \circ}_{1 / Q - 1} \bullet
  \underbrace{\circ \circ \cdots \circ}_{1 / Q - 1} \bullet
  \underbrace{\circ \circ \cdots \circ}_{1 / Q - 1} .
  \end{equation}
  Any attempt to reduce one of these spacings further necessarily increases
  another spacing.
\end{proof}

\begin{theorem}
  For any atomic ground state of the system, the largest sequence of unoccupied
  states cannot exceed $p$ sites.\label{1-CDWth2}
\end{theorem}

\begin{proof}
  Assume that there exists a ground state unit cell with a sequence of $(p + 1)$
  unoccupied sites, which we place to the very right in the cell by exploiting
  translational invariance. A periodic arrangement of these unit cells then
  takes the form
  \begin{equation}
  \underbrace{\begin{array}{|l|}
    \hline
    \underbrace{\bullet ? ? \cdots ? \bullet}_{\text{Block } A}
    \underbrace{\circ \circ \cdots \circ}_p \circ\\
    \hline
    \end{array} \cdots \begin{array}{|l|}
    \hline
    \underbrace{\bullet ? ? \cdots ? \bullet}_{\text{Block } A}
    \underbrace{\circ \circ \cdots \circ}_p \circ\\
    \hline
    \end{array}}_{q\text{ times}} ,
  \end{equation}
  and has an energy density $E_A/(N / Q)$ where $E_A$ is the energy of the
  denoted block $A$. We can now move the right-most unoccupied sites in each
  cell to the very end of this chain without changing the energy density,
  resulting in the rearranged configuration
  \begin{equation}
    \begin{array}{|l|}
    \hline
    \underbrace{\bullet ? ? \cdots ? \bullet}_{\text{Block } A}
    \underbrace{\circ \circ \cdots \circ}_p\\
    \hline
    \end{array} \cdots \begin{array}{|l|}
    \hline
    \underbrace{\bullet ? ? \cdots ? \bullet}_{\text{Block } A}
    \underbrace{\circ \circ \cdots \circ}_p  \underbrace{\circ \circ \cdots
      \circ}_q\\
    \hline
    \end{array} . \label{eq:1-CDWth2unitcell2}
  \end{equation}
  As guaranteed by the finite density $Q>1/{p+1}$,  block $A$ contains at least
  one particle that makes a finite contribution $E_{\Delta}$  to the energy  of
  this state. We can now take this particle from block $A$ and place it into the
  terminating segment of the chain,
  \begin{equation}
    \begin{array}{|l|}
    \hline
    \underbrace{\bullet ? ? \cdots ? \bullet}_{\text{Block } A}
    \underbrace{\circ \circ \cdots \circ}_p\\
    \hline
    \end{array} \cdots \begin{array}{|l|}
    \hline
    \underbrace{\bullet ? ? \cdots ? ?}_{\text{Block } A'}
    \underbrace{\circ \cdots \circ}_{p } \bullet \underbrace{\circ \circ
      \cdots \circ}_{q-1}\\
    \hline
    \end{array},
  \end{equation}
  where block $A'$ is block $A$ with the particle  replaced by a hole. Block
  $A'$ has energy $E_A - E_{\Delta}$, while the displaced particle no longer
  contributes to the energy of the state as soon as $q\geq p+1$, so that it is
  surrounded by at least $p$ unoccupied sites on both sides. This segment can
  now serve as a new unit cell with an energy density
  \begin{equation}
    \frac{pE_A - E_{\Delta}}{pN / Q} = \frac{E_A}{N / Q} -
    \frac{E_{\Delta}}{pN / Q},
  \end{equation}
  which lowers the energy in contradiction to our assumptions. A similar process
  can be used to show that a ground state cannot have a sequence of $(p + 2)$ or
  more unoccupied sites. Thus, we conclude that the largest spacing in any
  ground state has at most $p$ unoccupied sites.
\end{proof}

\begin{table}[t]
  \centering
  \caption{Charge-ordered ground states (GS) and their energies in the atomic limit of systems with commensurable particle densities
    $Q = 1 / (p - 1)$, where $p$ is the range of the interactions.
    The degeneracy of these states is denoted as $f$, which accounts for the translational freedom and the possibility of a mirror-reflected phase.
    $L_{\max}$ is the maximal size of the considered unit cell.
    \label{tab:1-CDWphases1}}
  \begin{tabular}{ccc}
    \hline \hline
    GS unit cell & Energy density & $f$ \\
    \hline
    \multicolumn{3}{c}{$p = 3, Q = 1/2, L_\text{max} = 28$}\\
    \hline
    $\bullet \circ$ & $ U_2/2$ & 2 \\
    $\bullet \bullet \circ \circ$ & $ (U_1 + U_3)/4$ &
    4 \\
    $\bullet \bullet \bullet \circ \circ \circ$ &
    $ (2 U_1 + U_2)/6$ & 6 \\
    \hline
    \multicolumn{3}{c}{$p = 4, Q = 1/3, L_\text{max} = 36$}\\
    \hline
    $\bullet \circ \circ$ & $ U_3/3$ & 3 \\
    $\bullet \bullet \circ \circ \circ \circ$ & $
    U_1/6$ & 6 \\
    $\bullet \circ \bullet \circ \circ \circ$ & $
    (U_2 + U_4)/6$ & 6 \\
    $\bullet \bullet \circ \circ \circ \bullet \circ \circ \circ$ &
    $ (U_1 + 2 U_4)/9$ & 9 \\
    $\bullet \circ \bullet \circ \bullet \circ \circ \circ \circ$ &
    $ (2 U_2 + U_4)/9$ & 9 \\
    $\bullet \circ \bullet \circ \circ \bullet \circ \bullet \circ \circ
    \circ \circ$ & $ (2 U_2 + U_3)/12$ & 12 \\
    $\bullet \bullet \bullet \circ \circ \circ \circ \bullet \circ \bullet
    \circ \circ \circ \circ \bullet \circ \bullet \circ \circ \circ \circ$ &
    $ (2 U_1 + 3 U_2)/21$ & $21$ \\
    \hline
    \multicolumn{3}{c}{$p = 5, Q = 1/4, L_\text{max} = 32$}\\
    \hline
    $\bullet \circ \circ \circ$ & $ U_4/4$ & 4 \\
    $\bullet \circ \circ \bullet \circ \circ \circ \circ$ &
    $ (U_3 + U_5)/8$ & 8 \\
    $\bullet \circ \bullet \circ \circ \circ \circ \circ$ &
    $ U_2/8$ & 8 \\
    $\bullet \circ \bullet \circ \circ \circ \circ \bullet \circ \circ
    \circ \circ$ & $ (U_2 + 2 U_5)/12$ & 12 \\
    $\bullet \circ \circ \bullet \circ \circ \bullet \circ \circ \circ
    \circ \circ$ & $ 2 U_3/12$ & 12 \\
    $\bullet \bullet \circ \circ \circ \circ \bullet \circ \circ \circ
    \circ \bullet \circ \circ \circ \circ$ & $ (U_1 + 3 U_5)/16$ & 16
    \\
    $\bullet \bullet \circ \circ \circ \circ \circ \bullet \bullet \circ
    \circ \circ \circ \circ \bullet \circ \circ \circ \circ \circ$ &
    $ 2 U_1/20$ & 20 \\
    \hline
    \multicolumn{3}{c}{$p = 6, Q = 1/5, L_\text{max} = 40$}\\
    \hline
    $\bullet \circ \circ \circ \circ$ & $U_5 / 5$ & 5\\
    $\bullet \circ \circ \circ \circ \circ \bullet \circ \circ \circ$ & $(U_4 +
    U_6) / 10$ & 10\\
    $\bullet \circ \circ \circ \circ \circ \circ \bullet \circ \circ$ & $U_3 /
    10$ & 10\\
    $\bullet \circ \circ \circ \circ \circ \bullet \circ \circ \bullet \circ
    \circ \circ \circ \circ$ & $(U_3 + 2 U_6) / 15$ & 15\\
    $\bullet \circ \circ \circ \circ \circ \circ \bullet \bullet \circ \circ
    \circ \circ \circ \circ$ & $U_1 / 15$ & 15\\
    $\bullet \circ \circ \circ \circ \circ \circ \bullet \circ \circ \circ
    \bullet \circ \circ \circ$ & $2 U_4 / 15$ & 15\\
    $\bullet \circ \circ \circ \circ \circ \bullet \circ \bullet \circ \circ
    \circ \circ \circ \bullet \circ \circ \circ \circ \circ$ & $(U_2 + 3 U_6) /
    20$ & 20\\
    $\bullet \circ \circ \circ \circ \circ \bullet \bullet \circ \circ \circ
    \circ \circ \bullet \circ \circ \circ \circ \circ \bullet \circ \circ \circ
    \circ \circ$ & $(U_1 + 4 U_6) / 25$ & 25\\
    $\bullet \circ \circ \circ \circ \circ \circ \bullet \circ \bullet \circ
    \circ \circ \circ \circ \circ \bullet \circ \bullet \circ \circ \circ \circ
    \circ \circ$ & $2 U_2 / 25$ & 25\\
    \hline \hline
  \end{tabular}
\end{table}

\begin{table*}[t]
  \centering
  \caption{Charge-ordered ground states as Table \ref{tab:1-CDWphases1}, but for systems with particle densities
    $Q = 1 / (p - 2)$.\label{tab:1-CDWphases2}}
  \begin{tabular}{ccc}
    \hline \hline
    GS unit cell & Energy density & $f$ \\
    \hline
    \multicolumn{3}{c}{$p = 4, Q = 1/2, L_\text{max} = 26$}\\
    \hline
    $\bullet \circ$ & $ (U_2 + U_4)/2$ & 2 \\
    $\bullet \bullet \circ \circ$ & $ (U_1 + U_3 + 2
    U_4)/4$ & 4 \\
    $\bullet \bullet \bullet \circ \circ \circ$ &
    $ (2 U_1 + U_2 + U_4)/6$ & 6 \\
    $\bullet \bullet \bullet \bullet \circ \circ \circ \circ$ &
    $ (3 U_1 + 2 U_2 + U_3)/8$ & 8 \\
    $\bullet \bullet \circ \bullet \circ \circ \bullet \circ$ &
    $ (U_1 + 2 U_2 + 3 U_3)/8$ & 8 \\
    \hline
    \multicolumn{3}{c}{$p = 5, Q = 1/3, L_\text{max} = 27$}\\
    \hline
    $\bullet \circ \circ$ & $ U_3/3$ & $3$ \\
    $\bullet \circ \bullet \circ \circ \circ$ & $
    (U_2 + U_4)/6$ & $6$ \\
    $\bullet \bullet \circ \circ \circ \circ$ & $
    (U_1 + U_5)/6$ & $6$ \\
    $\bullet \bullet \circ \circ \circ \bullet \circ \circ \circ$ &
    $ (U_1 + 2 U_4 + 2 U_5)/9$ & $9$ \\
    $\bullet \circ \bullet \circ \bullet \circ \circ \circ \circ$ &
    $ (2 U_2 + U_4 + U_5)/9$ & $9$ \\
    $\bullet \circ \bullet \circ \circ \bullet \circ \bullet \circ \circ
    \circ \circ$ & $ (2 U_2 + U_3 + 3 U_5)/12$ & $12$ \\
    $\bullet \bullet \circ \circ \bullet \circ \circ \circ \bullet \circ
    \circ \circ \bullet \circ \circ$ & $ (U_1 + 2 U_3 + 4 U_4)/15$ &
    $15$ \\
    $\bullet \bullet \bullet \circ \circ \circ \circ \circ \bullet \bullet
    \circ \circ \circ \circ \circ$ & $ (3 U_1 + U_2)/15$ & 15 \\
    $\bullet \bullet \circ \circ \bullet \bullet \circ \circ \circ \circ
    \circ \bullet \bullet \circ \circ \circ \circ \circ$ & $ (3
    U_1 + U_3 + 2 U_4 + U_5)/18$ & $18$ \\
    $\bullet \bullet \circ \circ \bullet \circ \circ \bullet \bullet \circ
    \circ \circ \circ \circ \bullet \bullet \circ \circ \circ \circ \circ$ &
    $ (3 U_1 + 2 U_3 + 2 U_4)/21$ & $21$ \\
    $\bullet \bullet \bullet \circ \circ \circ \circ \bullet \circ \bullet
    \circ \circ \circ \circ \bullet \circ \bullet \circ \circ \circ \circ$ &
    $ (2 U_1 + 3 U_2 + 3 U_5)/21$ & $21$ \\
    $\bullet \circ \bullet \circ \circ \circ \circ \circ \bullet \bullet
    \bullet \circ \circ \circ \circ \circ \bullet \bullet \bullet \circ \circ
    \circ \circ \circ$ & $ (4 U_1 + 3 U_2)/24$ & 24 \\
    \hline \hline
  \end{tabular}
\end{table*}

\begin{table*}[t]
  \centering
  \caption{Charge-ordered ground states as Table \ref{tab:1-CDWphases1}, but for systems at half filling ($Q = 1 / 2$).
  \label{tab:1-CDWphases3}}
  \begin{tabular}{ccc}
    \hline \hline
    GS unit cell & Energy density & $f$\\
    \hline
    \multicolumn{3}{c}{$p = 5, Q = 1/2, L_\text{max} = 26$}\\
    \hline
    $\bullet \circ$ & $ (U_2 + U_4)/2$ & 2\\
    $\bullet \bullet \circ \circ$ & $ (U_1 + U_3 + 2
    U_4 + U_5)/4$ & 4\\
    $\bullet \bullet \circ \bullet \circ \circ$ &
    $ (U_1 + U_2 + 2 U_3 + U_4 + U_5)/6$ & $2 \times 6$\\
    $\bullet \bullet \bullet \circ \circ \circ$ & $ (2 U_1 + U_2
    + U_4 + 2 U_5)/6$ & 6\\
    $\bullet \bullet \circ \bullet \circ \circ \bullet \circ$ &
    $ (U_1 + 2 U_2 + 3 U_3 + 3 U_5)/8$ & 8\\
    $\bullet \bullet \bullet \bullet \circ \circ \circ \circ$ &
    $ (3 U_1 + 2 U_2 + U_3 + U_5)/8$ & 8\\
    $\bullet \bullet \circ \bullet \bullet \circ \circ \bullet \circ \circ$
    & $ (2 U_1 + U_2 + 4 U_3 + 3 U_4)/10$ & 10\\
    $\bullet \bullet \bullet \bullet \bullet \circ \circ \circ \circ \circ$
    & $ (4 U_1 + 3 U_2 + 2 U_3 + U_4)/10$ & 10\\
    \hline
    \multicolumn{3}{c}{$p = 6, Q = 1/2, L_\text{max} = 26$}\\
    \hline
    $\bullet \circ$ & $ (U_2 + U_4 + U_6)/2$ & 2\\
    $\bullet \bullet \circ \circ$ & $ (U_1 + U_3 + 2
    U_4 + U_5)/4$ & 4\\
    $\bullet \bullet \circ \bullet \circ \circ$ &
    $ (U_1 + U_2 + 2 U_3 + U_4 + U_5 + 3 U_6)/6$ & $2 \times 6$\\
    $\bullet \bullet \bullet \circ \circ \circ$ & $ (2 U_1 + U_2
    + U_4 + 2 U_5 + 3 U_6)/6$ & 6\\
    $\bullet \circ \bullet \bullet \circ \bullet \circ \circ$ &
    $ (U_1 + 2 U_2 + 3 U_3 + 3 U_5 + 2 U_6)/8$ & 8\\
    $\bullet \bullet \bullet \bullet \circ \circ \circ \circ$ &
    $ (3 U_1 + 2 U_2 + U_3 + U_5 + 2 U_6)/8$ & 8\\
    $\bullet \bullet \circ \bullet \bullet \circ \circ \bullet \circ \circ$
    & $ (2 U_1 + U_2 + 4 U_3 + 3 U_4 + 3 U_6)/10$ & 10\\
    $\bullet \bullet \bullet \bullet \bullet \circ \circ \circ \circ \circ$
    & $ (4 U_1 + 3 U_2 + 2 U_3 + U_4 + U_6)/10$ & 10\\
    $\bullet \circ \bullet \bullet \circ \bullet \circ \bullet \circ \circ
    \bullet \circ$ & $ (U_1 + 4 U_2 + 3 U_3 + 2 U_4 + 5 U_5)/12$ &
    12\\
    $\bullet \bullet \bullet \bullet \bullet \bullet \circ \circ \circ
    \circ \circ \circ$ & $ (5 U_1 + 4 U_2 + 3 U_3 + 2 U_4 + U_5)/12$
    & 12\\
    $\bullet \bullet \bullet \circ \bullet \circ \circ \circ \bullet \circ
    \bullet \bullet \bullet \circ \circ \bullet \circ \circ$ & $
    (4 U_1 + 4 U_2 + 4 U_3 + 5 U_4 + 2 U_5 + 3 U_6)/18$ & 18\\
    $\bullet \bullet \bullet \circ \bullet \circ \circ \circ \bullet
    \bullet \circ \bullet \bullet \circ \circ \bullet \circ \circ$ &
    $ (4 U_1 + 3 U_2 + 5 U_3 + 5 U_4 + 2 U_5 + 3 U_6)/18$ & $2 \times
    18$\\
    \hline \hline
  \end{tabular}
\end{table*}

\subsection{Specific cases}\label{sec:insulatingPhasesList}

The general properties listed above allow us to significantly reduce the
effective charge-configuration space of ground-state candidates. Based on
Property 1 and exploiting the system's translational invariance, we can place
the guaranteed large spacing towards the front of the sequence, and therefore
fix the first $1 / Q$ sites to
\begin{equation}
  \bullet \underbrace{\circ \circ \cdots \circ}_{1 / Q - 1} .
  \label{eq:1-CDWth1fix}
\end{equation}
This reduces the effective configuration space to reduced systems of size $(N -
1) / Q$ and $(N - 1)$ particles. Based on Property 2, we then can remove any
charge configuration with unoccupied segments exceeding $p$, which at the same
time significantly reduces the maximal unit-cell size encountered in the
construction.

For each admissible state obtained in this way, we determined the general
expression of the ground-state energy density as a function of the interaction
parameters $\{ U_m \}$. Next, we discarded symbolically all configurations that
can never drop below the energy densities of all other charge configurations.
The final list contains the energy densities of all phases that can have the
lowest energy for some set of values $\{ U_m \}$. This leads to the
charge-ordered phases listed in the following tables.

Table \ref{tab:1-CDWphases1} lists the unit cells and energy densities for
critical densities $Q = 1 / (p - 1)$ and $p = 3$, $4$, $5$, $6$, with the
unit-cell size limited to the specified values $L_{\max}$. In these cases we are
highly confident that there are no ground states with larger unit cells. Table
\ref{tab:1-CDWphases2} presents the unit cells and energy densities for  $Q = 1
/ (p - 2)$ with $p = 4$ and 5. Notice that for $p = 5$ we find ground-state unit
cells of size up to $(L_{\max} - 3)$, so that we cannot fully exclude the
possibility of additional ground-state configurations with even larger unit
cells. Finally, results for $Q = 1 / 2, p = 5$ and $6$ are presented in Table
\ref{tab:1-CDWphases3}. Amongst the combinations listed in Table
\ref{tab:insulatingphases}, this leaves the case $p=6$, $Q=1/3$ where we find 63
phases with $L\leq L_{\max}=27$, and $p=6$, $Q=1/4$, where we find 23 phases
with $L\leq L_{\max}=32$, which defines the limit of our computational
capabilities; the corresponding phases are therefore not listed here.

In all these tables, the degeneracy of the states accounts for the translational
displacement by a finite number of sites (up to the size of the unit cell),  as
well as for the possible duplication by a distinct mirror-reflected phase.


%

\end{document}